\documentclass[10ps]{article}
\usepackage{amsmath}
\usepackage{amsfonts}
\usepackage{amssymb}
\usepackage{amscd}
\usepackage{amsthm}
\usepackage{authblk}
\usepackage{color}
\usepackage{enumitem}
\usepackage{fullpage}
\usepackage{latexsym}
\usepackage{mathrsfs}
\usepackage{picture}
\usepackage{stmaryrd}
\usepackage{tikz}
\usepackage{tikz-cd}
\usepackage[colorinlistoftodos]{todonotes}
\usepackage[all]{xy}

\newcommand{\Sup}{\bigvee}

\newcommand{\id}{\mathit{id}}
\newcommand{\bool}{\mathbb{B}}
\newcommand{\scomp}{\,\textcolor{orange}{\bullet}\,}
\newcommand{\pcomp}{\,\textcolor{teal}{\bullet}\,}
\newcommand{\sconv}{\,\textcolor{orange}{\ast}\,}
\newcommand{\pconv}{\,\textcolor{teal}{\ast}\,}
\newcommand{\srel}{\textcolor{orange}{R}}
\newcommand{\prel}{\textcolor{teal}{R}}
\newcommand{\se}{\textcolor{orange}{1}}
\newcommand{\pe}{\textcolor{teal}{1}}
\newcommand{\sE}{\textcolor{orange}{E}}
\newcommand{\pE}{\textcolor{teal}{E}}
\newcommand{\sD}{\textcolor{orange}{D}}
\newcommand{\pD}{\textcolor{teal}{D}}
\newcommand{\sdot}{\textcolor{orange}{\otimes}}
\newcommand{\pdot}{\textcolor{teal}{\otimes}}
\newcommand{\sid}{\textcolor{orange}{\mathit{id}_E}}
\newcommand{\pid}{\textcolor{teal}{\mathit{id}_E}}

\newcommand{\pstar}{\textcolor{teal}{\star}}
\newcommand{\sstar}{\textcolor{orange}{\star}}
\newcommand{\ai}[1]{$\textrm{I}_{#1}$}
\newcommand{\ri}[1]{$\textrm{RI}_{#1}$}
\newcommand{\pow}{\mathcal{P}}
\newcommand{\At}{\mathit{At}}

\newtheorem{theorem}{Theorem}[section]
\newtheorem{proposition}[theorem]{Proposition}
\newtheorem{lemma}[theorem]{Lemma}
\newtheorem{corollary}[theorem]{Corollary}

\theoremstyle{definition}
\newtheorem{definition}[theorem]{Definition}
\newtheorem{example}[theorem]{Example}
\newtheorem{question}[theorem]{Question}

\definecolor{jccolor}{rgb}{1,0.5,0.5}
\definecolor{sdcolor}{cmyk}{1,0,1,0}
\definecolor{gscolor}{cmyk}{1,0,0,0}

\begin{document}

\title{Convolution and Concurrency}

\author{James Cranch}
\author{Simon Doherty}
\author{Georg Struth}
\affil{University of Sheffield\\ United Kingdom}

\date{}

\maketitle

\begin{abstract}
  We show how concurrent quantales and concurrent Kleene algebras
  arise as convolution algebras $Q^X$ of functions from structures $X$
  with two ternary relations that satisfy relational interchange laws
  into concurrent quantales or Kleene algebras $Q$. The elements of
  $Q$ can be understood as weights; the case $Q=\bool$ corresponds to
  a powerset lifting. We develop a correspondence theory between
  relational properties in $X$ and algebraic properties in $Q$ and
  $Q^X$ in the sense of modal and substructural logics,  and boolean
  algebras with operators. As examples, we construct the concurrent
  quantales and Kleene algebras of $Q$-weighted words, digraphs,
  posets, isomorphism classes of finite digraphs and pomsets.
\end{abstract}

\pagestyle{plain}


\section{Introduction}\label{S:introduction}

Our initial motivation has been the provision of recipes for
constructing graph models for concurrent quantales and concurrent
Kleene algebras~\cite{HMSW11}.  These axiomatise the sequential and
parallel compositions $\cdot$ and $\|$ of concurrent and distributed
systems as well as their finite sequential and parallel iterations and
impose that these compositions interact via a lax interchange law
$(w\| x) \cdot (y\| z)\le (w\cdot y)\| (x\cdot z)$.  Two classical
models are languages over finite words with respect to sequential and
shuffle composition in interleaving concurrency, and languages over
partial orders or partial words (pomsets) with respect to serial and
parallel composition in true concurrency. The relation $\le$ is
interpreted as set inclusion in these models.

In both models, the language-level algebras are constructed by lifting
structural properties of compositions from single objects---single
words, single posets---to power sets. In fact, both liftings are
instances of the classical Stone-type duality between $n+1$-ary
relations and $n$-ary operators (or modalities) on power set boolean
algebras~\cite{JonssonT51,Goldblatt}; here for ternary relations and
binary operators. In the word model, the ternary relations on words
are $u=v\cdot w$ and $u\in v\| w$; the binary operators on power sets
are $X\cdot Y=\{u\cdot v\mid u\in X\land v\in Y\}$ and
$X\| Y = \bigcup\{u\|v\mid u\in Y\land v\in Y\}$. In the poset model,
the ternary relations on posets are $P=P_1\cdot P_2$ provided
$P_1\cdot P_2$ is defined ($P_1\cdot P_2$ being disjoint union with
additional arrows from each element of $P_1$ to each element of $P_2$)
and $P\preceq P_1\| P_2$ provided $P_1\| P_2$ is defined ($P_1\| P_2$
being disjoint union and the relation holds if there is a bijective
order morphism from the right-hand poset to the left-hand one).  The
binary operators on power sets are
$X\cdot Y=\{P_1\cdot P_2\mid P_1\in X, P_2\in Y \text{ and } P_1\cdot
P_2 \text{ is defined}\}$
and
$X\| Y=\{P \mid \exists P_1 \in X,P_2\in Y.\ P\preceq P_1\| P_2\text{
  and } P_1\| P_2 \text{ is defined}\}$.

Both constructions generalise further to weighted words and weighted
po(m)sets~\cite{Handbook} and beyond that---yet ignoring
interchange---to arbitrary functions $X\to Q$ from partial monoids or
ternary relations over $X$ into quantales
$Q$~\cite{DongolHS16,DongolHS17}. The binary operations on function
spaces $Q^X$ then generalise to convolutions of the form
$(f\ast g)\, x = \Sup\{f\, x \bullet g\, y\mid R^x_{yz}\}$, and the
algebra on the function space $Q^X$ is called \emph{convolution
  algebra}.  This raises the more specific question how concurrent
quantales and similar structures on $Q^X$ can be constructed from
ternary relations on $X$ and weight quantales $Q$---in particular the
above lax interchange law and its variants on $Q^X$. This question is
not only of structural interest. Operationally, checking relational
properties on $X$ tends to be much simpler than those on $Q^X$, and
the first activity suffices if the construction of $Q^X$ from $X$ and
$Q$ is uniform. The rest of this article investigates this question.

First, in Section~\ref{S:summary}, we summarise the previous approach
to relational convolution in $Q^X$~\cite{DongolHS17}, where $X$ is a
set equipped with a ternary relation and $Q$ a quantale, and introduce
the basic lifting construction, namely that $Q^X$ forms a quantale if
$X$ satisfies a relational associativity law and $Q$ is a quantale,
and if a suitable set of relational units is present.

In Section~\ref{S:correspondence-interchange} we prove novel
correspondence results between relational interchange laws on $X$ and
algebraic interchange laws on $Q$ and
$Q^X$. Proposition~\ref{P:correspondence1} shows that interchange on
$X$ and $Q$ give rise to those on $Q^X$. In addition, under mild
non-degeneracy conditions on $Q$, interchange laws on $Q$ and $Q^X$
give rise to those on $X$ (Proposition~\ref{P:correspondence2}) and,
under mild non-degeneracy conditions on $X$, interchange laws on $X$
and $Q^X$ give rise to those on $Q$
(Proposition~\ref{P:correspondence3}).  In combination, these results
show that interchange laws on $X$ and $Q$ are precisely what is needed
to obtain such laws on $Q^X$.

Additional correspondences are then presented in
Section~\ref{S:further-correspondences}. First, we prove such results
for sets of relational units in $X$ and quantalic units on $Q$ and
$Q^X$ and show how the above degeneracy conditions simplify in the
presence of units. Secondly, we show how correspondences for
(semi-)associativity and commutativity laws arise from those for
interchange.

Equipped with these correspondences we then introduce relational
interchange monoids and interchange quantales in
Section~\ref{S:interchange-quantales} and package the individual
correspondences for these structures in the main theorem of this
article (Theorem~\ref{P:interchange-quantale-correspondence}): a
correspondence result between relational monoids $X$, interchange
quantales $Q$ and interchange quantales $Q^X$. Interchange quantales
are essentially concurrent quantales without commutativity assumptions
on the ``parallel'' composition. In addition we prove a weak
Eckmann-Hilton argument that shows that certain small interchange laws
are subsumed by the one presented above.

In  light of the duality between $n+1$-ary relations and boolean
algebras with $n$-ary operators, the natural question arises how a
more general duality between $X$, $Q$ and $X^Q$ can be
obtained. Partial results are already known~\cite{HardingWW18}.  We
explain the special case of the power set lifting ($Q=\bool$) in
Section~\ref{S:duality} and relate this results with constructions
from Section~\ref{S:correspondence-interchange}, but leave the general
case for future work.

In Section~\ref{S:interchange-kas} we specialise
Theorem~\ref{P:interchange-quantale-correspondence} to interchange
Kleene algebras and concurrent Kleene algebras, which requires
finiteness and grading assumptions on ternary relations
(Theorem~\ref{P:interchange-ka-correspondence}).

Finally, Sections~\ref{S:shuffle}-\ref{S:graph-type-languages} apply
the constructions from
Theorem~\ref{P:interchange-quantale-correspondence} and
\ref{P:interchange-ka-correspondence} to the examples mentioned above.
In Section~\ref{S:shuffle} we construct the concurrent quantale and
Kleene algebra of $Q$-weighted shuffle languages using an isomorphism
between ternary relations and certain
multimonoids~\cite{GalmicheL06}. In Section~\ref{S:graph-languages}
and \ref{S:graph-type-languages} we construct the concurrent quantale
and Kleene algebra of $Q$-weighted digraph languages and those of
isomorphism classes of finite digraphs. To prepare for these
constructions, Section~\ref{S:pims} introduces partial interchange
monoids, which form relational interchange monoids under certain
restrictions.  It then suffices to show that graphs under the
operations $\cdot$ and $\|$ outlined form such monoids.  The
specialisation to (weighted) partial orders and partial words or
pomsets, which are isomorphism classes of labelled partial orders, is
then straightforward.

Ultimately these results yield uniform construction principles for
(weighted) concurrent quantales and Kleene algebras from simpler
structures such as ternary relations, multimonoids and similar ordered
monoidal structures: to construct such structures it suffices to know
the underlying relational structure, the rest is then
automatic. Beyond that, our results provide valuable structural
insights that might be stepping stones to future duality results.


\section{Relational Convolution: a Summary} \label{S:summary}

This section summarises the general approach.

Relational convolution~\cite{DongolHS17} has its origins in J\'onsson
and Tarski's boolean algebras with operators~\cite{JonssonT51}, Rota's
foundations of combinatorics~\cite{Rota64}, Sch\"utzenberger's
approach to language theory~\cite{BerstelReutenauer,Handbook} and
Goguen's L-fuzzy maps and relations~\cite{Goguen67}. It is an
operation in the algebra of functions $X\to Q$ from a set $X$ into a
complete lattice $Q$ equipped with an additional operation $\bullet$
of composition and constrained by a ternary relation on $X$, which we
identify with a predicate of type $X\to X\to X\to \bool$.  It is defined as
\begin{equation*}
  (f\ast g)\, x = \Sup_{y,z:R^x_{yz}} f\, y\bullet g\, z,
\end{equation*}
where the right-hand side abbreviates
$\Sup \{f\, y \bullet g\, z\mid R^x_{yz}\}$ and $\Sup$ denotes the
supremum in $Q$. It is well known that the function space $Q^X$ forms
a complete lattice when the order and sups in $Q$ are extended
pointwise~\cite{AbramskyJung}.  Yet the convolution $\ast$ need not satisfy any algebraic
laws on $Q^X$, unless conditions on $R$ and $Q$ are imposed.

This is reminiscent of modal correspondence theory, where conditions
on relational Kripke frames force algebraic properties of modal
operators and vice versa, or more generally to dualities between
categories of $n+1$-ary relational structures and those of boolean
algebras with $n$-ary operators~\cite{JonssonT51,Goldblatt}. In
fact, $R$ is a ternary relational structure and $\ast$ a binary
modality similar to the product of the Lambek calculus~\cite{Lambek58},
the chop operation of interval temporal logics~\cite{MM83} or the
separating conjunction of separation logic~\cite{OHRY01}.

\begin{example}\label{ex:incidence}
  Let $X$ be an incidence algebra of closed intervals (over
  $\mathbb{R}$, say)~\cite{Rota64}, with interval fusion $[p,q][r,s]$
  equal to $[p,s]$ if $q=r$, and undefined otherwise. Let $R^x_{yz}$
  hold if the fusion of intervals $y$ and $z$ is defined and equal to
  $x$.  Let $Q=\bool$ be the (complete) lattice of booleans with
  $\bullet$ as meet.  Functions $X\to \bool$ are then predicates
  ranging over intervals in $X$. The predicate $f\ast g$ holds of an
  interval $x$ whenever it can be decomposed into a prefix interval
  $y$ and a suffix interval $z$ by fusion such that $f\, y$ and
  $g\, z$ both hold.  This captures the semantics of the binary chop
  modality of interval temporal logics~\cite{MM83} .\qed
\end{example}

It is well known that chop is associative in the convolution algebra
$\mathbb{B}^X$ due to associativity of meet in $\mathbb{B}$ and
associativity of interval fusion in $X$---up to definedness.

\begin{definition}[\cite{Rosenthal90}]\label{D:quantale}
A \emph{quantale} $Q$ is a complete lattice equipped
with an associative composition $\bullet$ that preserves sups in both
arguments: for all $a,b\in Q$ and
$A,B\subseteq Q$,
\begin{equation*}
  a\bullet \left(\Sup B\right) = \Sup \{a\bullet b\mid b\in
  B\}\qquad\text{ and }\qquad
  \left(\Sup A\right) \bullet b = \Sup \{a\bullet b\mid a\in A\}.
\end{equation*}
A quantale is \emph{unital} if $\bullet$ has a unit $1$.
\end{definition}
Convolution is then associative in $Q^X$ if $R$ is \emph{relationally
  associative}~\cite{DongolHS17}: for all $x,u,v,w\in X$,

\begin{equation*}
  \exists y.\ R^y_{uv} \wedge R^x_{yw} \Leftrightarrow \exists
  y.\ R^x_{uy} \wedge R^y_{vw}.
\end{equation*}
This yields one direction of a correspondence between the ternary
relation $R$ and convolution $\ast$ viewed as a binary modality.
Similarly, convolution is commutative in $Q^X$ if $R$ is
\emph{relationally commutative}: for all $x,u,v\in X$,
\begin{equation*}
  R^x_{uv} \Rightarrow R^x_{vu}.
\end{equation*}
Finally, if the unary relation $E^x$ is a set of relational units for
$R$ and the quantale $Q$ unital with unit of composition $1$, then
convolution has the indicator function $\id_E$ as a left and right
unit, where
\begin{equation*}
  \id_E\, x =
  \begin{cases}
    1,& \text{if } E^x,\\
0, & \text{otherwise}.
  \end{cases}
\end{equation*}
\begin{definition}[\cite{DongolHS17}]\label{D:rel-units}
  The set $E\subseteq X$ is a \emph{set of relational units} for the
  ternary relation $R$ over $X$ if it satisfies, for all $x,y\in X$,
\begin{gather*}
\exists e.\  R^x_{ex} \wedge E^e, \qquad
R^x_{ey} \wedge E^e \Rightarrow x=y,\qquad
\exists e.\  R^x_{xe}\wedge E^e, \qquad
R^x_{ye}\wedge E^e\Rightarrow x = y.
\end{gather*}
\end{definition}
\noindent This guarantees that each $x\in X$ has a unique left unit as
well as a unique right one.  With the Kronecker delta function
$\delta_x: X\to \bool$ defined as $\delta_x\, y$ equal to $1$ if $x=y$
and to $0$ otherwise, therefore,
\begin{equation*}
\id_E = \Sup_{e:E^e} \delta_e.
\end{equation*}

The convolution algebras on $Q^X$ can now be described as follows.
\begin{theorem}[\cite{DongolHS17}]\label{P:conv-algebras} ~
\begin{enumerate}
\item If $R$ is relationally associative and $Q$ a quantale, then
  $Q^X$ is a quantale.
\item If $R$ is relationally associative and commutative and $Q$
  an abelian quantale, then $Q^X$ is an abelian quantale.
\item If $R$ is relationally associative and has relational
  units, and $Q$ is a unital quantale, then $Q^X$ is a unital
  quantale
.
\end{enumerate}
\end{theorem}

\begin{example}\label{ex:languages}
  Let $(X^\ast, \cdot,\varepsilon)$ be the free monoid over $X$. For
  $u,v,w\in X^\ast$ define $R^u_{vw} \Leftrightarrow u = v\cdot w$.
  Then $R$ is relationally associative with relational unit
  $\varepsilon$.  For any quantale $Q$, the convolution algebra is the
  quantale $Q^{X^\ast}$ of $Q$-weighted languages over $X$ whereas
  $\bool^{X^\ast}$ is the usual language quantale over
  $X$. Convolution is (weighted) language product (the boolean case is
  also known as complex or Minkowski product). The construction
  generalises to arbitrary monoids.

  Words are \emph{finitely decomposable} in that each word can only be
  split into finitely many prefix/suffix pairs.  All sups in
  convolutions therefore remain finite and $Q$ can be replaced by an
  arbitrary semiring.  This yields the usual weighted languages
  formalised as rational power series in the sense of Sch\"utzenberger
  and Eilenberg~\cite{BerstelReutenauer,Handbook}. \qed
\end{example}

\begin{example}\label{ex:relations}
  Define a composition $\cdot :X\times X\to X$ on a set $X$ such that
  $(a,b)\cdot (c,d)$ is equal to $(a,d)$ if $b=c$ and undefined
  otherwise.  For $x,y,z\in X\times X$, let $R^x_{yz}$ hold if and
  only if $y\cdot z$ is defined and equal to $x$, and let
  $E=\{(a,a)\mid a \in X\}$.  Then $R$ is relationally associative and
  the elements of $E$ are the relational units.  For any quantale $Q$,
  the convolution algebra $Q^{X\times X}$ is the quantale of
  $Q$-weighted binary relations over $X$, while $\bool^{X\times X}$ is
  simply the quantale of binary relations over $X$.  Convolution is
  (weighted) relational composition. This specialises to quantales of
  weighted closed intervals in linear orders, as in
  Example~\ref{ex:incidence}, which can be represented by ordered
  pairs $(a,b)$ in which $a\le b$.\qed
\end{example}

\begin{example}\label{ex:sep-logic}
  Define a composition $\oplus$ on the set of partial functions of
  type $X\rightharpoonup Y$ such that $f \oplus g$ is $f\cup g$ if
  $\mathit{dom}\, f$ and $\mathit{dom}\, g$ are disjoint and undefined
  otherwise.  The set $Y^X$ can model the heap memory area with
  addresses in $X$, values in $Y$ and $\oplus$ as heaplet
  addition. For $f,g,h:X\rightharpoonup Y$ let $R^f_{gh}$ hold if and
  only if $g\oplus h$ is defined and equal to $f$.  Then $R$ is
  relationally associative and commutative; the empty partial function
  (which is undefined everywhere) is its relational unit. For any
  abelian quantale $Q$, the convolution algebra is the abelian
  quantale $Q^{(Y^X)}$ of $Q$-weighted assertions of separation logic
  over the set $Y^X$ of heaps. Convolution is separating
  conjunction~\cite{DongolHS16}. The standard assertion algebra of separation
  logic is formed by $\bool^{(Y^X)}$.\qed
\end{example}


\section{Correspondences for Interchange
  Laws}\label{S:correspondence-interchange}

Theorem~\ref{P:conv-algebras} generalises to correspondences between
quantales with two compositions $\textcolor{orange}{\bullet}$ and
$\textcolor{teal}{\bullet}$ related by seven interchange laws and
relational structures with suitable relational constraints. The choice
of the interchange laws is explained in
Section~\ref{S:interchange-quantales}; the six small interchange laws
are precisely those that can be derived from the seventh in the
presence of suitable units. We start with the relational structures.

\begin{definition}\label{D:rel-bi-magma}~
\begin{enumerate}
\item A \emph{relational magma} $(X,R)$ is a set $X$ equipped with a
  ternary relation $R:X\to X\to X\to \bool$.  It is \emph{unital} if
  there is a set $E$ of relational units satisfying the axioms in
  Definition~\ref{D:rel-units}.
\item A \emph{relational bi-magma} $(X,\srel,\prel)$ is a set $X$
  equipped with two ternary relations $\srel$ and $\prel$. It is
  \emph{unital} if $\srel$ has a set of relational units $\sE$ and
  $\prel$ a set of relational units $\pE$.
\end{enumerate}
\end{definition}

\noindent The constraints considered on a bi-magma $X$ are, for
$t,u,v,w,x,y,z\in X$, the \emph{relational interchange laws}
\begin{align}
\srel^x_{uv} &\Rightarrow \prel^x_{uv}, \tag{\ri1}  \label{eq:ri1}\\
\srel^x_{uv} &\Rightarrow \prel^x_{vu}, \tag{\ri2}  \label{eq:ri2}\\
\exists y.\ \srel^x_{uy} \wedge \prel^y_{vw} &\Rightarrow
\exists y.\ \srel^y_{uv}  \wedge \prel^x_{yw},
\tag{\ri3}\label{eq:ri3}\\
\exists y.\ \prel^y_{uv}  \wedge \srel^x_{yw} &\Rightarrow  \exists y.\ \prel^x_{uy}\wedge \srel^y_{vw},
\tag{\ri4}\label{eq:ri4}\\
\exists y.\ \srel^x_{uy} \land \prel^y_{vw} &\Rightarrow \exists
                                             y. \prel^x_{vy} \land
                                             \srel^y_{uw},\tag{\ri5}\label{eq:ri5}\\
\exists y.\ \prel^y_{uv}\land \srel^x_{yw} & \Rightarrow \exists
                                            y. \srel^y_{uw}\land \prel^x_{yv},\tag{\ri6}\label{eq:ri6}\\
\exists y,z.\ \prel^y_{tu} \wedge \srel^x_{yz} \wedge
\prel^z_{vw} &\Rightarrow \exists y,z.\ \srel^y_{tv} \wedge \prel^x_{yz} \wedge \srel^z_{uw}.
\tag{\ri7}\label{eq:ri7}
\end{align}
They memoise the relationships between the trees in $X$
shown in
Figure~\ref{fig:memotrees}.
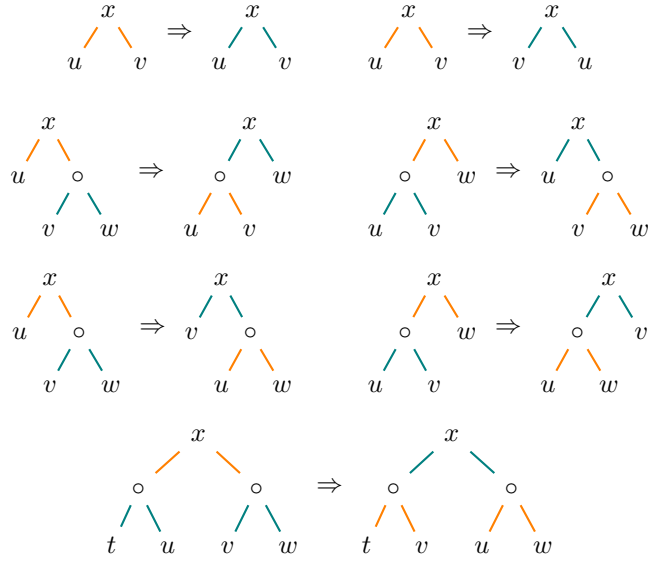
\begin{figure}[t]
  \centering
\begin{equation*}
    \begin{tikzcd}[column sep=-.06cm, row sep=-.06cm]
& x \arrow[dash,line width=.8pt,orange,ddl]\arrow[dash,line width=.8pt,orange,ddr] &\\
&\phantom{x}&\\
u&&v
\end{tikzcd}
\Rightarrow
    \begin{tikzcd}[column sep= -.06cm, row sep=-.06cm]
& x \arrow[dash,line width=.8pt,teal,ddl]\arrow[dash,line width=.8pt,teal,ddr] &\\
&\phantom{x}&\\
u&&v
\end{tikzcd} 
\qquad
    \begin{tikzcd}[column sep=-.06cm, row sep=-.06cm]
& x \arrow[dash,line width=.8pt,orange,ddl]\arrow[dash,line width=.8pt,orange,ddr] &\\
&\phantom{x}&\\
u&&v
\end{tikzcd}
\Rightarrow
\begin{tikzcd}[column sep= -.06cm, row sep=-.06cm]
& x \arrow[dash,line width=.8pt,teal,ddl]\arrow[dash,line width=.8pt,teal,ddr] &\\
&\phantom{x}&\\
v&&u
\end{tikzcd} 
\end{equation*}
\begin{equation*}
\begin{tikzcd}[column sep= -.1cm, row sep=.3cm]
&x\ar[dash,line width=.8pt,orange,dl]\ar[dash,line width=.8pt,orange,dr]&&\\
u&&\circ \ar[dash,line width=.8pt,teal,dl]\ar[dash,line width=.8pt,teal,dr]&\\
&v &&w
\end{tikzcd}
\Rightarrow
\begin{tikzcd}[column sep= -.1cm, row sep=.3cm]
&&x \ar[dash,line width=.8pt,teal,dl]\ar[dash,line width=.8pt,teal,dr]&\\
&\circ \ar[dash,line width=.8pt,orange,dl]\ar[dash,line width=.8pt,orange,dr]&& w\\
u&&v&
\end{tikzcd} 
\qquad
\begin{tikzcd}[column sep= -.1cm, row sep=.3cm]
&&x \ar[dash,line width=.8pt,orange,dl]\ar[dash, line width=.8pt,orange,dr]&\\
&\circ \ar[dash,line width=.8pt,teal,dl]\ar[dash,line width=.8pt,teal,dr]&& w\\
u&&v&
\end{tikzcd}
\Rightarrow
\begin{tikzcd}[column sep= -.1cm, row sep=.3cm]
&x\ar[dash,line width=.8pt,teal,dl]\ar[dash,line width=.8pt,teal,dr]&&\\
u&&\circ\ar[dash,line width=.8pt,orange,dl]\ar[dash,line width=.8pt,orange,dr]&\\
&v &&w
\end{tikzcd} 
\end{equation*}
\begin{equation*}
\begin{tikzcd}[column sep= -.1cm, row sep=.3cm]
&x\ar[dash,line width=.8pt,orange,dl]\ar[dash,line width=.8pt,orange,dr]&&\\
u&&\circ \ar[dash,line width=.8pt,teal,dl]\ar[dash,line width=.8pt,teal,dr]&\\
&v &&w
\end{tikzcd}
\Rightarrow
\begin{tikzcd}[column sep= -.1cm, row sep=.3cm]
&x\ar[dash,line width=.8pt,teal,dl]\ar[dash,line width=.8pt,teal,dr]&&\\
v&&\circ \ar[dash,line width=.8pt,orange,dl]\ar[dash,line width=.8pt,orange,dr]&\\
&u &&w
\end{tikzcd}
\qquad
\begin{tikzcd}[column sep= -.1cm, row sep=.3cm]
&&x \ar[dash,line width=.8pt,orange,dl]\ar[dash, line width=.8pt,orange,dr]&\\
&\circ \ar[dash,line width=.8pt,teal,dl]\ar[dash,line width=.8pt,teal,dr]&& w\\
u&&v&
\end{tikzcd}
\Rightarrow
\begin{tikzcd}[column sep= -.1cm, row sep=.3cm]
&&x \ar[dash,line width=.8pt,teal,dl]\ar[dash, line width=.8pt,teal,dr]&\\
&\circ \ar[dash,line width=.8pt,orange,dl]\ar[dash,line width=.8pt,orange,dr]&& v\\
u&&w&
\end{tikzcd}
 \end{equation*}
\begin{equation*}
\begin{tikzcd}[column sep= -.1cm, row sep=.3cm]
&&&x \ar[dash,line width=.8pt,orange,dll]\ar[dash, line width=.8pt,orange,drr]&&&\\
&\circ \ar[dash,line width=.8pt,teal,dl]\ar[dash,line width=.8pt,teal,dr]&&&& \circ \ar[dash,line width=.8pt,teal,dl]\ar[dash,line width=.8pt,teal,dr]&\\
t&& u && v && w
\end{tikzcd}
\Rightarrow
\begin{tikzcd}[column sep= -.1cm, row sep=.3cm]
&&&x \ar[dash,line width=.8pt,teal,dll]\ar[dash, line width=.8pt,teal,drr]&&&\\
&\circ \ar[dash,line width=.8pt,orange,dl]\ar[dash,line width=.8pt,orange,dr]&&&& \circ \ar[dash,line width=.8pt,orange,dl]\ar[dash,line width=.8pt,orange,dr]&\\
t&& v && u && w
\end{tikzcd}
\end{equation*}
 \caption{Trees memoised by the relational interchange laws.}
  \label{fig:memotrees}
\end{figure}

Next we turn to quantales. As their monoidal structure emerges in the
constructions, we generalise.

\begin{definition}\label{D:pre-bi-quantale}~
\begin{enumerate}
\item A \emph{prequantale}~\cite{Rosenthal90} is a structure
  $(Q,\le,\bullet)$ such that $(Q,\le)$ is a complete lattice and the
  binary operation $\bullet$ on $Q$ preserves sups in both
  arguments. It is \emph{unital} if $\bullet$ has unit $1$.
\item  A \emph{bi-prequantale} is a structure $(Q,\le,\scomp,\pcomp)$ such
  that $(Q,\le,\scomp)$ and $(Q,\le,\pcomp)$ are both prequantales. It
  is unital if $\scomp$ has unit $\se$ and $\pcomp$ unit $\pe$.
\end{enumerate}
\end{definition}

\noindent A quantale is thus a prequantale with associative composition.

For
$a,b,c,d\in Q$ we define the \emph{algebraic interchange laws}
\begin{align}
  a\scomp b &\le a\pcomp b,\tag{\ai1}\label{eq:i1}\\
 a\scomp b &\le b\pcomp c,\tag{\ai2}\label{eq:i2}\\
a\scomp (b\pcomp c) &\le (a\scomp b)\pcomp c,\tag{\ai3}\label{eq:i3}\\
  (a\pcomp b)\scomp c &\le a\pcomp (b\scomp
                        c),\tag{\ai4}\label{eq:i4}\\
a\scomp (b\pcomp c) &\le b\pcomp (a\scomp c),\tag{\ai5}\label{eq:i5}\\
(a\pcomp b)\scomp c &\le (a\scomp c)\pcomp b,\tag{\ai6}\label{eq:i6}\\
 (a\pcomp b)\scomp (c\pcomp d) &\le (a\scomp c)\pcomp (b\scomp
 d).\tag{\ai7}\label{eq:i7}
\end{align}
Interestingly, the syntax trees of these laws in $Q$, as shown in
Figure~\ref{fig:syntaxtrees}, have the structure of the trees in $X$ in
Figure~\ref{fig:memotrees}. The following example provides some
intuition.
\begin{figure}[t]
  \centering
\begin{equation*}
    \begin{tikzcd}[column sep=-.06cm, row sep=-.06cm]
& \scomp \arrow[dash,ddl]\arrow[dash,ddr] &\\
&\phantom{x}&\\
a&&b
\end{tikzcd}
\le
    \begin{tikzcd}[column sep= -.06cm, row sep=-.06cm]
& \pcomp \arrow[dash,ddl]\arrow[dash,ddr] &\\
&\phantom{x}&\\
a&&b
\end{tikzcd} 
\qquad
    \begin{tikzcd}[column sep=-.06cm, row sep=-.06cm]
& \scomp \arrow[dash,ddl]\arrow[dash,ddr] &\\
&\phantom{x}&\\
a&&b
\end{tikzcd}
\le
\begin{tikzcd}[column sep= -.06cm, row sep=-.06cm]
& \pcomp \arrow[dash,ddl]\arrow[dash,ddr] &\\
&\phantom{x}&\\
b&&a
\end{tikzcd} 
\end{equation*}
\begin{equation*}
\begin{tikzcd}[column sep= -.1cm, row sep=.3cm]
&\scomp\ar[dash,dl]\ar[dash,dr]&&\\
a&&\pcomp \ar[dash,dl]\ar[dash,dr]&\\
&b &&c
\end{tikzcd}
\le
\begin{tikzcd}[column sep= -.1cm, row sep=.3cm]
&&\pcomp \ar[dash,dl]\ar[dash,dr]&\\
&\scomp
 \ar[dash,dl]\ar[dash,dr]&& c\\
a&&b&
\end{tikzcd} 
\qquad
\begin{tikzcd}[column sep= -.1cm, row sep=.3cm]
&&\scomp \ar[dash,dl]\ar[dash,dr]&\\
&\pcomp \ar[dash,dl]\ar[dash,dr]&& c\\
a&&b&
\end{tikzcd}
\le
\begin{tikzcd}[column sep= -.1cm, row sep=.3cm]
&\pcomp\ar[dash,dl]\ar[dash,dr]&&\\
a&&\scomp\ar[dash,dl]\ar[dash,dr]&\\
&b &&c
\end{tikzcd} 
\end{equation*}
\begin{equation*}
\begin{tikzcd}[column sep= -.1cm, row sep=.3cm]
&\scomp\ar[dash,dl]\ar[dash,dr]&&\\
a&&\pcomp \ar[dash,dl]\ar[dash,dr]&\\
&b &&c
\end{tikzcd}
\le
\begin{tikzcd}[column sep= -.1cm, row sep=.3cm]
&\pcomp\ar[dash,dl]\ar[dash,dr]&&\\
b&&\scomp\ar[dash,dl]\ar[dash,dr]&\\
&a &&c
\end{tikzcd}
\qquad
\begin{tikzcd}[column sep= -.1cm, row sep=.3cm]
&&\scomp \ar[dash,dl]\ar[dash,dr]&\\
&\pcomp \ar[dash,dl]\ar[dash,dr]&& c\\
a&&b&
\end{tikzcd}
\le
\begin{tikzcd}[column sep= -.1cm, row sep=.3cm]
&&\pcomp \ar[dash,dl]\ar[dash,dr]&\\
&\scomp \ar[dash,dl]\ar[dash,dr]&& b\\
a&&c&
\end{tikzcd}
 \end{equation*}
\begin{equation*}
\begin{tikzcd}[column sep= -.1cm, row sep=.3cm]
&&&\scomp \ar[dash,dll]\ar[dash,drr]&&&\\
&\pcomp \ar[dash,dl]\ar[dash,dr]&&&& \pcomp \ar[dash,dl]\ar[dash,
dr]&\\
a&& b && c && d
\end{tikzcd}
\le
\begin{tikzcd}[column sep= -.1cm, row sep=.3cm]
&&&\pcomp \ar[dash,dll]\ar[dash,drr]&&&\\
&\scomp \ar[dash,dl]\ar[dash,dr]&&&& \scomp\ar[dash,dl]\ar[dash,dr]&\\
a&& c && b && d
\end{tikzcd}
\end{equation*}
 \caption{Syntax trees of the algebraic interchange laws.}
  \label{fig:syntaxtrees}
\end{figure}
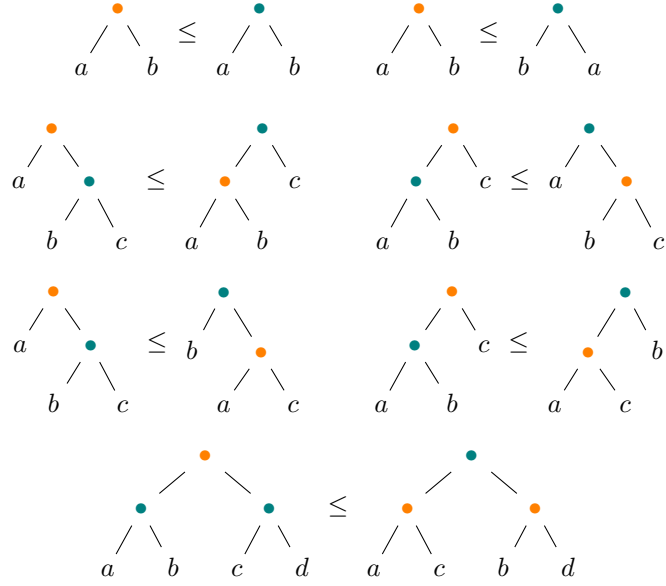
\begin{example}\label{ex:memotrees}
  Let $X=Q$, $\srel^x_{ab} \Leftrightarrow x\le a\scomp b$ and
  $\prel^x_{ab}  \Leftrightarrow x\le a\pcomp b$. Then, for instance,
\begin{equation*}
\exists y,z.\ \prel^y_{ab} \wedge \srel^x_{yz} \wedge
  \prel^z_{cd} \Leftrightarrow x \le (a\pcomp b)\scomp (c\pcomp d),
\end{equation*}
and likewise for the other terms in interchange laws. The relational
and algebraic interchange laws then translate into each other. For
instance, for (\ref{eq:ri7}) and (\ref{eq:i7}),
\begin{align*}
  \left(\exists y,z.\ \prel^y_{ab} \wedge \srel^x_{yz} \wedge
  \prel^z_{cd} \Rightarrow \exists y,z.\ \srel^y_{ac} \wedge
  \prel^x_{yz} \wedge \srel^z_{bd}\right)
&\Leftrightarrow
\left(x \le (a\pcomp b)\scomp (c\pcomp d) \Rightarrow x\le (a\scomp
  c)\pcomp (b\scomp d)\right)\\
&\Leftrightarrow
(a\pcomp b)\scomp (c\pcomp d) \le (a\scomp c)\pcomp (b\scomp d),
\end{align*}
that is,
\begin{equation*}
\left(
\begin{tikzcd}[column sep= -.2cm, row sep=.2cm]
&&&x \ar[dash,line width=.8pt,orange,dll]\ar[dash, line width=.8pt,orange,drr]&&&\\
&\circ \ar[dash,line width=.8pt,teal,dl]\ar[dash,line width=.8pt,teal,dr]&&&& \circ \ar[dash,line width=.8pt,teal,dl]\ar[dash,line width=.8pt,teal,dr]&\\
t&& u && v && w
\end{tikzcd}
\Rightarrow
\begin{tikzcd}[column sep= -.2cm, row sep=.2cm]
&&&x \ar[dash,line width=.8pt,teal,dll]\ar[dash, line width=.8pt,teal,drr]&&&\\
&\circ \ar[dash,line width=.8pt,orange,dl]\ar[dash,line width=.8pt,orange,dr]&&&& \circ \ar[dash,line width=.8pt,orange,dl]\ar[dash,line width=.8pt,orange,dr]&\\
t&& v && u && w
\end{tikzcd}
\right)
\ \Leftrightarrow \
\left(
\begin{tikzcd}[column sep= -.2cm, row sep=.2cm]
&&&\scomp \ar[dash,dll]\ar[dash,drr]&&&\\
&\pcomp \ar[dash,dl]\ar[dash,dr]&&&& \pcomp \ar[dash,dl]\ar[dash,dr]&\\
t&& u && v && w
\end{tikzcd}
\le
\begin{tikzcd}[column sep= -.2cm, row sep=.2cm]
&&&\pcomp \ar[dash,dll]\ar[dash,drr]&&&\\
&\scomp \ar[dash,dl]\ar[dash,dr]&&&& \scomp \ar[dash,dl]\ar[dash,dr]&\\
t&& v && u && w
\end{tikzcd}
\right).
\end{equation*}
With this particular encoding, the relational interchange laws
simply memoise the algebraic ones. \qed
\end{example}
In general, however, the relationship between relational and algebraic
convolution is more complex. The
left-to-right translation in Example~\ref{ex:memotrees} can therefore
fail.

For functions $f,g:X\to Q$ we define the convolutions
\begin{equation*}
  (f\sconv g)\, x = \Sup_{y,z:\srel^x_{yz}} f\, y\scomp g\,
  z\quad\text{ and}\quad (f\pconv g)\, x = \Sup_{y,z:\prel^x_{yz}} f\, y\pcomp g\, z.
\end{equation*}
They can be represented using trees in $X$, $Q$ and $Q^X$ as
\begin{equation*}
\begin{tikzcd}[column sep= -.2cm, row sep=.2cm]
&\sconv \ar[dash,dl]\ar[dash,dr]&\\
f && g
\end{tikzcd}
=
\lambda x.\ \Sup
\left\{
\begin{tikzcd}[column sep= -.2cm, row sep=.2cm]
&\scomp \ar[dash,dl]\ar[dash,dr]&\\
f\,y && g\, z
\end{tikzcd}
\, \middle|\,
\begin{tikzcd}[column sep= -.2cm, row sep=.2cm]
&x \ar[dash,line width=.8pt,orange,dl]\ar[dash, line width=.8pt,orange,dr]&\\
y && z
\end{tikzcd}
\right\}
\qquad\text{ and }\qquad
\begin{tikzcd}[column sep= -.2cm, row sep=.2cm]
&\pconv \ar[dash,dl]\ar[dash,dr]&\\
f && g
\end{tikzcd}
=
\lambda x.\ \Sup
\left\{
\begin{tikzcd}[column sep= -.2cm, row sep=.2cm]
&\pcomp \ar[dash,dl]\ar[dash,dr]&\\
f\,y && g\, z
\end{tikzcd}
\, \middle|\,
\begin{tikzcd}[column sep= -.2cm, row sep=.2cm]
&x \ar[dash,line width=.8pt,teal,dl]\ar[dash, line width=.8pt,teal,dr]&\\
y && z
\end{tikzcd}
\right\}.
\end{equation*}
Convolution thus translates trees with the same structure in $X$ and
$Q$ into trees in $Q^X$.

One can then prove correspondences between relational and
algebraic interchange laws. First we show that relational interchange
laws in $X$ and algebraic interchange laws in $Q$ force algebraic
interchange laws in the convolution algebra on $Q^X$.

\begin{proposition}\label{P:correspondence1}
  Let $X$ be a relational bi-magma and $Q$ a bi-prequantale. Then
  (\ri{k}) in $X$ and (\ai{k}) in $Q$ imply (\ai{k}) in $Q^X$, for
  each $1\leq k\leq 7$.
\end{proposition}
\begin{proof} Let
  $\exists y,z.\ \prel^y_{tu} \wedge \srel^x_{yz} \wedge
    \prel^z_{vw}\Rightarrow \exists y,z.\ \srel^y_{tv} \wedge
    \prel^x_{yz} \wedge \srel^z_{uw}$ in $X$
  and
  $(w\pcomp x)\scomp (y\pcomp z) \le (w\scomp y)\pcomp (x\scomp z)$ in
  $Q$. Then
   \begin{align*}
     \left((f\pconv g) \sconv (h\pconv k)\right)\, x
&= \Sup \left\{ \Sup \left\{f\, t\pcomp g\, u \mid \prel^y_{tu}\right\}
  \scomp \Sup \left\{h\, v\pcomp k\, w\mid \prel^z_{vw}\right\} \,
  \middle|\, \srel^x_{yz}\right\}\\
&= \Sup\left\{ (f\,
  t \pcomp g\, u) \scomp (h\, v \pcomp k\, w)\,\middle|\, \exists y,z.\
  \prel^y_{tu}\wedge \srel^x_{yz}\wedge \prel^z_{vw}\right\}\\
&\le \Sup\left\{ (f\,
  t \scomp h\, v) \pcomp (g\, u \scomp k\, w)\, \middle|\, \exists y,z.\
  \srel^y_{tv}\wedge \prel^x_{yz}\wedge \srel^z_{uw}\right\}\\
&= \Sup \left\{ \Sup \left\{f\, t\scomp h\, v \mid \srel^y_{tv}\right\}
  \pcomp \Sup \left\{g\, u\scomp k\, w\mid \srel^z_{uw}\right\} \,
  \middle|\, \prel^x_{yz}\right\}\\
&= \left((f\sconv h)\pconv (g\sconv k)\right)\, x
   \end{align*}
   proves (\ref{eq:i7}) in $Q^X$. Alternatively, using trees,
\begin{eqnarray*}
\begin{tikzcd}[column sep= -.2cm, row sep=.2cm]
&&&\sconv \ar[dash,dll]\ar[dash,drr]&&&\\
&\pconv \ar[dash,dl]\ar[dash,dr]&&&& \pconv \ar[dash,dl]\ar[dash,dr]&\\
f&& g && h && k\end{tikzcd}
&=&
\lambda x.\ \Sup
\left\{
\begin{tikzcd}[column sep= -.2cm, row sep=.2cm]
&&&\scomp \ar[dash,dll]\ar[dash,drr]&&&\\
&\pcomp \ar[dash,dl]\ar[dash,dr]&&&& \pcomp \ar[dash,dl]\ar[dash,dr]&\\
f\, t&& g\, u && h\, v && k\, w
\end{tikzcd}
\, \middle|\,
\begin{tikzcd}[column sep= -.2cm, row sep=.2cm]
&&&x \ar[dash,line width=.8pt,orange,dll]\ar[dash, line width=.8pt,orange,drr]&&&\\
&\circ \ar[dash,line width=.8pt,teal,dl]\ar[dash,line width=.8pt,teal,dr]&&&& \circ \ar[dash,line width=.8pt,teal,dl]\ar[dash,line width=.8pt,teal,dr]&\\
t&& u && v && w
\end{tikzcd}
\right\} \\
&\le&
\lambda x.\ \Sup
\left\{
\begin{tikzcd}[column sep= -.2cm, row sep=.2cm]
&&&\pcomp \ar[dash,dll]\ar[dash,drr]&&&\\
&\scomp \ar[dash,dl]\ar[dash,dr]&&&& \scomp \ar[dash,dl]\ar[dash,dr]&\\
f\, t&& h\, v && g\, u && k\, w
\end{tikzcd}
\, \middle|\,
\begin{tikzcd}[column sep= -.2cm, row sep=.2cm]
&&&x \ar[dash,line width=.8pt,teal,dll]\ar[dash, line width=.8pt,teal,drr]&&&\\
&\circ \ar[dash,line width=.8pt,orange,dl]\ar[dash,line width=.8pt,orange,dr]&&&& \circ \ar[dash,line width=.8pt,orange,dl]\ar[dash,line width=.8pt,orange,dr]&\\
t&& v && u && w
\end{tikzcd}
\right\}\\
&=&
\begin{tikzcd}[column sep= -.2cm, row sep=.2cm]
&&&\pconv \ar[dash,dll]\ar[dash,drr]&&&\\
&\sconv \ar[dash,dl]\ar[dash,dr]&&&& \sconv \ar[dash,dl]\ar[dash,dr]&\\
f&& h && g && k\end{tikzcd}.
\end{eqnarray*}
The proofs for the small interchange laws are similar, and left to the
reader. In particular, the proof of (\ref{eq:i3}) from (\ref{eq:ri3})
and that of (\ref{eq:i4}) from (\ref{eq:ri4}) are related by
\emph{opposition}: one can be obtained from the other by swapping the
operands of $\sconv$, $\pconv$, $\scomp$ and $\pcomp$ and the lower
indices of $\srel$ and $\prel$, that is, by reversing the algebraic
syntax trees in $Q$ and the trees in $X$ memoised in the relational
interchange laws. The same holds for the proof of (\ref{eq:i5}) from
(\ref{eq:ri5}) and that of (\ref{eq:i6}) from (\ref{eq:ri6}).
\end{proof}

Next we show that, under mild nondegeneracy conditions on
$Q$, algebraic interchange laws in $Q^X$ force relational
interchange laws in $X$, and that under mild nondegeneracy
conditions on $X$, algebraic interchange laws in $Q^X$ force algebraic
interchange laws in $Q$.  Yet we begin with a definition and some
lemmas.

For all $x,y\in X$ and $a\in Q$, we define the function
$\delta^a_x:X\to Q$ by
\begin{equation*}
  \delta^a_x\, y =
\begin{cases}
  a, & \text{ if } x= y,\\
0 & \text{ otherwise}
\end{cases}
\end{equation*}
and the function $(-\mid -):Q\to\bool\to Q$, for all $a\in Q$ and
$P:\bool$, by
\begin{equation*}
  (a\mid P) =
  \begin{cases}
    a, &\text{ if } P,\\
0, & \text{ otherwise}.
  \end{cases}
\end{equation*}
Obviously, $\delta^a_x\, y = (a\mid x = y)$.
\begin{lemma}\label{P:delta-conv-props}
Let $X$ be a relational bi-magma and $Q$ a bi-prequantale. For all
$a,b,c,d\in Q$ and $x,t,u,v,w\in X$,
\begin{enumerate}[label=(\arabic*)]
\item $\left(\delta^a_u \sconv \delta^b_v\right) x = \left(a \scomp b\mid
  \srel^x_{uv}\right)$,
\item $\left(\delta^a_u \sconv \left(\delta^b_v \pconv
      \delta^c_w\right)\right) x = \left(a
  \scomp (b \pcomp c) \mid
  \exists y.\ \srel^x_{uy}\wedge \prel^y_{vw}\right) $,
\item $\left(\left(\delta^a_u
  \pconv \delta^b_v\right)\sconv \delta^c_w\right) x = \left((a \pcomp b)\scomp c\mid
  \exists y.\ \prel^y_{uv}\wedge \srel^x_{yw}\right) $,
\item $\left(\left(\delta^a_t\pconv \delta^b_u\right)\sconv \left(\delta^c_v \pconv
  \delta^d_w\right)\right) x = \left((a\pcomp b)\scomp (c\pcomp d)\mid \exists y,z.\
  \prel^y_{tu}\wedge \srel^x_{yz}\wedge \prel^z_{vw}\right) $,
\item properties (1)--(4) hold with colours interchanged.
\end{enumerate}
\end{lemma}
\begin{proof}
  For (4), we use the proof of Proposition~\ref{P:correspondence1} to
  calculate
  \begin{align*}
    ((\delta^a_t\pconv \delta^b_u)\sconv (\delta^c_v \pconv
  \delta^d_w))\, x
&= \Sup\{(\delta^a_t\, x_1 \pcomp \delta^b_u\,
  x_2)\scomp (\delta^c_v\, x_3 \pcomp \delta^d_w\, x_4)\mid \exists
  y,z. \ \prel^y_{x_1x_2} \wedge \srel^x_{yz}\wedge
  \prel^z_{x_3x_4}\}\\
&= ((a\pcomp b)\scomp (c\pcomp d)\mid \exists y,z.\
  \prel^y_{tu}\wedge \srel^x_{yz}\wedge \prel^z_{vw})\\
&\le (a\pcomp b)\scomp (c\pcomp d).
  \end{align*}
Alternatively, using trees,
\begin{eqnarray*}
\begin{tikzcd}[column sep= -.2cm, row sep=.2cm]
&&&\sconv \ar[dash,dll]\ar[dash,drr]&&&\\
&\pconv \ar[dash,dl]\ar[dash,dr]&&&& \pconv \ar[dash,dl]\ar[dash,dr]&\\
\delta^a_t&& \delta^b_u && \delta^c_v&& \delta^d_w\end{tikzcd}
&=&
\lambda x.\ \Sup
\left\{
\begin{tikzcd}[column sep= -.2cm, row sep=.2cm]
&&&\scomp \ar[dash,dll]\ar[dash,drr]&&&\\
&\pcomp \ar[dash,dl]\ar[dash,dr]&&&& \pcomp \ar[dash,dl]\ar[dash,dr]&\\
\delta^a_t\, x_1 && \delta^b_u\, x_2 && \delta^c_v\, x_3 &&
\delta^d_w\, x_4
\end{tikzcd}
\, \middle|\,
\begin{tikzcd}[column sep= -.2cm, row sep=.2cm]
&&&x \ar[dash,line width=.8pt,orange,dll]\ar[dash, line width=.8pt,orange,drr]&&&\\
&\circ \ar[dash,line width=.8pt,teal,dl]\ar[dash,line width=.8pt,teal,dr]&&&& \circ \ar[dash,line width=.8pt,teal,dl]\ar[dash,line width=.8pt,teal,dr]&\\
x_1&& x_2 && x_3 && x_4
\end{tikzcd}
\right\}\\
&=&
\lambda x.\
\left(
\begin{tikzcd}[column sep= -.2cm, row sep=.2cm]
&&&\scomp \ar[dash,dll]\ar[dash,drr]&&&\\
&\pcomp \ar[dash,dl]\ar[dash,dr]&&&& \pcomp \ar[dash,dl]\ar[dash,dr]&\\
a && b && c && d
\end{tikzcd}
\, \middle|\,
\begin{tikzcd}[column sep= -.2cm, row sep=.2cm]
&&&x \ar[dash,line width=.8pt,orange,dll]\ar[dash, line width=.8pt,orange,drr]&&&\\
&\circ \ar[dash,line width=.8pt,teal,dl]\ar[dash,line width=.8pt,teal,dr]&&&& \circ \ar[dash,line width=.8pt,teal,dl]\ar[dash,line width=.8pt,teal,dr]&\\
t&& u && v && w
\end{tikzcd}
\right)\\
&\le &
\begin{tikzcd}[column sep= -.2cm, row sep=.2cm]
&&&\scomp \ar[dash,dll]\ar[dash,drr]&&&\\
&\pcomp \ar[dash,dl]\ar[dash,dr]&&&& \pcomp \ar[dash,dl]\ar[dash,dr]&\\
a && b && c && d
\end{tikzcd}.
\end{eqnarray*}

All other proofs are similar, and left to the reader. In particular,
(3) follows from (2) by opposition.
\end{proof}

The next lemma shows that the trees memoised by the relational
interchange laws can be expressed in terms of deltas and convolution
in the presence of the following mild nondegeneracy conditions on $Q$:
\begin{gather}
\exists a, b \in Q.\ a \scomp b \neq 0,  \tag{$D_1$}\label{eq:d1}\\
\exists a,b,c \in Q.\ a\scomp (b\pcomp c) \neq 0,
\tag{$D_2$}\label{eq:d2}\\
\exists a,b,c \in Q.\ (a\pcomp b)\scomp c \neq 0,
\tag{$D_3$}\label{eq:d3}\\
\exists a,b,c,d \in Q.\ (a\pcomp b)\scomp (c\pcomp d) \neq 0. \tag{$D_4$}\label{eq:d4}
\end{gather}
\begin{lemma}\label{P:delta-rel}
Let $X$ be a relational bi-magma and $Q$ a
bi-prequantale. Then
\begin{enumerate}[label=(\arabic*)]
 \item $\srel^x_{yz} \Rightarrow \left(\delta_y^a \sconv
  \delta_z^b\right) x = a\scomp b$, and the converse implication follows from (\ref{eq:d1}),
\item
  $\exists y.\ \srel^x_{uy}\wedge \prel^y_{vw} \Rightarrow
  \left(\delta^a_u \sconv \left(\delta^b_v \pconv
      \delta^c_w\right)\right) x = a\scomp (b\pcomp c)$,
  and the converse implication follows from (\ref{eq:d2}),
\item
  $\exists y.\ \prel^y_{uv}\wedge \srel^x_{yw} \Rightarrow \left(\left(\delta^a_u \pconv\delta^b_v\right) \sconv
    \delta^c_w\right) x = (a\pcomp
  b)\scomp c$,
  and the converse implication follows from (\ref{eq:d3}),
\item
  $\exists y,z.\ \prel^y_{tu}\wedge \srel^x_{yz}\wedge \prel^z_{vw}
  \Rightarrow
  \left(\left(\delta^a_t\pconv \delta^b_u\right)\sconv
    \left(\delta^c_v \pconv \delta^d_w\right)\right) x = (a\pcomp b)\scomp(c\pcomp d)$,
  and the converse implication follows from (\ref{eq:d4}),
\item properties (1)--(5) hold with colours interchanged, including in the
  degeneracy conditions.
\end{enumerate}
\end{lemma}
\begin{proof}

  For (4), suppose
  $\exists y,z.\ \prel^y_{tu}\wedge \srel^x_{yz}\wedge \prel^z_{vw}$.
Then
\begin{equation*}
  \left(\left(\delta^a_t\pconv \delta^b_u\right)\sconv
    \left(\delta^c_v \pconv \delta^d_w\right)\right) x =
\left((a\pcomp b)\scomp (c\pcomp d)\mid \exists y,z.\
  \prel^y_{tu}\wedge \srel^x_{yz}\wedge \prel^z_{vw}\right) = (a\pcomp b)\scomp(c\pcomp d)
\end{equation*}
by Lemma~\ref{P:delta-conv-props}(4).
For the converse implication,
  \begin{equation*}
    0\neq (a\pcomp b)\scomp(c\pcomp d) = \left((\delta^a_t\pconv
  \delta^b_u)\scomp (\delta^c_v \pconv \delta^d_w)\right)\, x = \left((a\pcomp
  b)\scomp(c\pcomp d) \mid \exists y,z.\
  \prel^y_{tu}\wedge \srel^x_{yz}\wedge \prel^z_{vw}\right)
  \end{equation*}
by Lemma~\ref{P:delta-conv-props}(4) and therefore $\exists y,z.\
  \prel^y_{tu}\wedge \srel^x_{yz}\wedge \prel^z_{vw}$.

All other proofs are similar, and left to the reader. (3)
follow from (2) by opposition.
\end{proof}

Intuitively, convolutions of delta functions represent trees in $X$ in
the function space $Q^X$ by
creating their ``shadows'' in $Q$---which requires nondegeneracy. The
case of Lemma~\ref{P:delta-rel}(4) and its dual are shown in
Figure~\ref{fig:delta-diagram}.

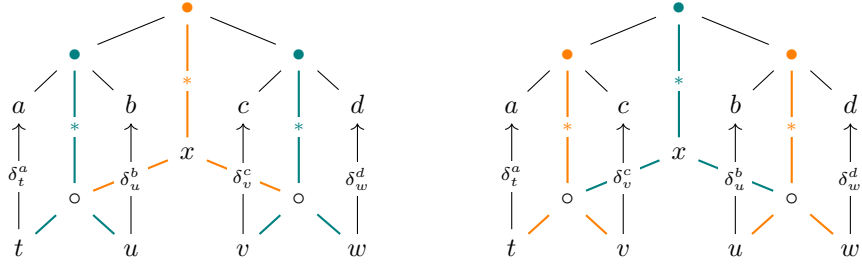
\begin{figure}[t]
  \centering

\begin{tikzcd}[column sep= .2cm, row sep=.2cm]
&&&\scomp \ar[dash,dll]\ar[dash,drr]&&&\\
&\pcomp \ar[dash,dl]\ar[dash,dr]&&&& \pcomp \ar[dash,dl]\ar[dash,dr]&\\
a && b && c && d\\
&&&x\ar[dash,line width=.8pt,orange,uuu, "\sconv" description] \ar[dash,line width=.8pt,orange,dll]\ar[dash, line width=.8pt,orange,drr]&&&\\
&\circ\ar[dash,line width=.8pt,teal,uuu, "\pconv" description] \ar[dash,line
width=.8pt,teal,dl]\ar[dash,line width=.8pt,teal,dr]&&&&
\circ\ar[dash,line width=.8pt,teal,uuu, "\pconv" description] \ar[dash,line width=.8pt,teal,dl]\ar[dash,line width=.8pt,teal,dr]&\\
t\arrow[uuu, "\delta^a_t" description]&& u\ar[crossing over, uuu,
"\delta^b_u" description] && v\ar[crossing over, uuu, "\delta^c_v"
description] && w\ar[uuu, "\delta^d_w" description]
\end{tikzcd}
\qquad\qquad
\begin{tikzcd}[column sep= .2cm, row sep=.2cm]
&&&\pcomp \ar[dash,dll]\ar[dash,drr]&&&\\
&\scomp \ar[dash,dl]\ar[dash,dr]&&&& \scomp \ar[dash,dl]\ar[dash,dr]&\\
a && c && b && d\\
&&&x\ar[dash,line width=.8pt,teal,uuu, "\pconv" description] \ar[dash,line width=.8pt,teal,dll]\ar[dash, line width=.8pt,teal,drr]&&&\\
&\circ\ar[dash,line width=.8pt,orange,uuu, "\sconv" description] \ar[dash,line
width=.8pt,orange,dl]\ar[dash,line width=.8pt,orange,dr]&&&&
\circ\ar[dash,line width=.8pt,orange,uuu, "\sconv" description] \ar[dash,line width=.8pt,orange,dl]\ar[dash,line width=.8pt,orange,dr]&\\
t\arrow[uuu, "\delta^a_t" description]&& v\ar[crossing over, uuu,
"\delta^c_v" description] && u\ar[crossing over, uuu, "\delta^b_u"
description] && w\ar[uuu, "\delta^d_w" description]
\end{tikzcd}

\caption{$\left(\delta^a_t\pconv \delta^b_u\right)\sconv
  \left(\delta^c_v \pconv \delta^d_w\right)\, x$
  representing
  $\exists y,z.\ \prel^y_{tu}\wedge \srel^x_{yz}\wedge \prel^z_{vw}$
  using $(a\pcomp b)\scomp(c\pcomp d)$ in
  Lemma~\ref{P:delta-rel}(4) together with its dual.}
  \label{fig:delta-diagram}
\end{figure}

We are now prepared to prove our second correspondence result, namely
that algebraic interchange laws in $Q^X$ force relational
interchange laws in $X$ subject to mild nondegeneracy conditions on $Q$.
\begin{proposition}\label{P:correspondence2}
  Let $X$ be a relational bi-magma and $Q$ a bi-prequantale. Then
  $(D_{\lceil\frac{k}{2}\rceil})$ in $Q$ and $(I_k)$ in $Q^X$ imply $(RI_k)$ in $X$, for
  each $1\le k\le 7$.
\end{proposition}
\begin{proof}
  Suppose $(a\pcomp b)\scomp (c\pcomp d) \neq 0$ for some
  $a,b,c,d\in Q$ and
  $\left(\delta^a_t\pconv \delta^b_u\right)\sconv \left(\delta^c_v\pconv \delta^d_w\right)
  \le \left(\delta^a_t\sconv \delta^c_v\right)\pconv \left(\delta^b_u\sconv
  \delta^d_w\right)$.
  Then, using Lemma~\ref{P:delta-rel}(4),
\begin{align*}
\exists y,z.\ \prel^y_{tu} \wedge \srel^x_{yz} \wedge \prel^z_{vw}
&\Leftrightarrow 0 \neq \left(\left(\delta^a_t\pconv \delta^b_u\right)\sconv
 \left(\delta^c_v\pconv \delta^d_w\right)\right) x\\
&\Rightarrow 0 \neq \left(\left(\delta^a_t\sconv \delta^c_v\right)\pconv \left(\delta^b_u\sconv
  \delta^d_w\right)\right) x\\
&\Leftrightarrow \exists y,z.\ \srel^y_{tv} \wedge \prel^x_{yz}
  \wedge \srel^z_{uw}
\end{align*}
proves $(RI_7)$. The remaining proofs are similar. Those for $(RI_3)$
and $(RI_4)$ and those for $(RI_5)$ and $(RI_6)$ are related by
opposition.\qedhere
\end{proof}
Finally, we prove the third correspondence result for interchange
laws, namely that algebraic interchange laws on $Q^X$ force those on
$Q$, subject to the following mild degeneracy conditions on $X$:
\begin{gather}
\exists x,u,v.\ \srel^x_{uv}, \tag{$RD_1$} \label{eq:rd1}\\
\exists x,y,u,v,w.\ \srel^x_{uy} \wedge \prel^y_{vw},
\tag{$RD_2$}\label{eq:rd2}\\
\exists x,y,u,v,w.\ \prel^y_{uv} \wedge \srel^x_{yw},
\tag{$RD_3$}\label{eq:rd3}\\
\exists x,y,z,t,u,v,w.\ \prel^y_{tu} \wedge \srel^x_{yz} \wedge
\prel^z_{vw}.
\tag{$RD_4$}\label{eq:rd4}
\end{gather}
\begin{proposition}\label{P:correspondence3}.
  Let $X$ be a relational bi-magma and $Q$ a bi-prequantale. Then
  $(RD_{\lceil\frac{k}{2}\rceil})$ in $X$ and $(I_k)$ in $Q^X$ imply
  $(I_k)$ in $Q$, for each $1\le k\le 7$.

\end{proposition}
\begin{proof}
  Suppose
  $(\delta^a_t\pconv \delta^b_u)\sconv (\delta^c_v\pconv \delta^d_w)
  \le (\delta^a_t\sconv \delta^c_v)\pconv (\delta^b_u\sconv
  \delta^d_w)$
  for some $a,b,c,d\in Q$ and let
  $\exists y,z.\ \prel^y_{tu} \wedge \srel^x_{yz} \wedge \prel^z_{vw}$
  for some $t,u,v,w\in X$. Then, using Lemma~\ref{P:delta-rel}(4) and
  \ref{P:delta-conv-props}(4),
\begin{equation*}
(a\pcomp b)\scomp (c\pcomp d)
= ((\delta^a_t \pconv\delta^b_u) \sconv (\delta^c_v\pconv
  \delta^d_w))\, x
\le ((\delta^a_t\sconv \delta^c_v)\pconv (\delta^b_u\sconv
  \delta^d_w))\, x
\le (a\scomp c)\pcomp (b\scomp d)
\end{equation*}
proves $(I_7)$ in $Q$.  The remaining
proofs are similar. Those for $(RD_3)$ and $(RD_4)$ and those for
$(RD_5)$ and $(RD_6)$ are related by opposition.
\end{proof}

It may be helpful to check the proofs of
Propositions~\ref{P:correspondence2} and \ref{P:correspondence3} with
the diagrams in Figure~\ref{fig:delta-diagram}.  The degeneracy
conditions are necessary. Indeed, if $Q$ is a singleton set, then so
is $Q^X$ and hence will obey all axioms independently of
$X$. Similarly, if all products on $Q$ vanish, then $Q^X$ will
automatically satisfy many axioms as all convolutions will be trivial.


\section{Further Correspondences}\label{S:further-correspondences}

When the relational bi-magma $X$ and the bi-prequantale $Q$ are both
unital, units can be defined in $Q^X$ as in
Section~\ref{S:summary}:
\begin{align*}
\sid\, x =
  \begin{cases}
    \se, & \text{if $\sE^x$},\\
0, & \text{otherwise}
  \end{cases}
\qquad\text{ and}\qquad
 \pid\, x=
  \begin{cases}
    \pe, & \text{if $\pE^x$},\\
0, & \text{otherwise}.
  \end{cases}
\end{align*}
Theorem~\ref{P:conv-algebras} then shows that $Q^X$ is a unital quantale if
$Q$ is a unital quantale and both compositions are associative and
unital in $X$.  We
restate the three kinds of correspondences for units in the weaker
setting of relational magmas and prequantales.
\begin{proposition}\label{P:correspondence-units}
Let $X$ be a relational magma and $Q$ a prequantale.
\begin{enumerate}
\item If $X$ and $Q$ are unital, then so is $Q^X$.
\item If $Q^X$ is unital and $1\neq 0$ in $Q$, then so is $X$.
\item If $Q^X$ is unital and $E\neq \emptyset$ in $X$ , then so is
  $Q$.
\end{enumerate}
\end{proposition}
\begin{proof}~
\begin{enumerate}
\item If $X$ and $Q$ are unital, then
  \begin{align*}
(f\ast \id_E)\, x
=\Sup \{f\, y \bullet \delta_e\, z \mid R^x_{yz} \land E^e\}
=\Sup \{f\, y \bullet 1 \mid \exists e.\ R^x_{ye} \land E^e\}
=(f\, x \mid \exists e.\ R^x_{xe} \land E^e)
= f\, x,
\end{align*}
where the last two steps use the relational unit axioms from
Definition~\ref{D:rel-units}. The proof for left units follows by
opposition.

\item If $\id_E$ is the right unit in $Q^X$, then
  \begin{equation*}
(1\mid (y=x))
=   \delta_y\, x
= (\delta_y\ast \id_E)\, x
= (1\mid \exists e.\ R^x_{ye} \land E^e).
  \end{equation*}
  Suppose $1\neq 0$.  Then $x=y$ implies
  $\exists e.\ R^x_{xe} \land E^e$, the existence axiom for right
  relational units, and $\exists e.\ R^x_{ye} \land E^e$
  implies that $x=y$, the uniqueness axiom.  The proofs for left units
  follow by opposition.

 \item If $\id_E$ is the right unit in $Q^X$, then
   \begin{equation*}
     a\bullet 1
= (a\bullet 1\mid \exists e.\ R^x_{xe} \land E^e)
=\Sup\{\delta_x\, x \mid \exists e.\ R^x_{xe} \land E^e\}
=  (\delta^a_x \ast \id_E)\, x
= \delta^a_x\, x
= a
   \end{equation*}
   proves that $1$ is a right unit in $Q$. The left unit law follows
   by opposition.\qedhere
\end{enumerate}
\end{proof}

In the presence of non-trivial units in $X$ and $Q$, the nondegeneracy
conditions for interchange laws in Proposition~\ref{P:correspondence2}
and \ref{P:correspondence3} simplify. Condition (\ref{eq:d1}) becomes
trivial with $\se\scomp\se = \se\neq 0$, condition (\ref{eq:d2}) with
$\se\scomp (\pe\pcomp\pe) = \pe\neq 0$ and condition (\ref{eq:d3}) by
opposition. Condition (\ref{eq:d4}) reduces to $(\pe\pcomp \pe)\scomp
(\pe \pcomp \pe) = \pe \scomp \pe \neq 0$, but remains non-trivial. It
becomes trivial when $\se=\pe$.  It is easy to check that the
nondegeneracy conditions (\ref{eq:rd1})--(\ref{eq:rd3}) become trivial
in a similar way, using the fact that $\srel^e_{ee}$ holds for all
$e\in \sE$ and $\prel^e_{ee}$ for all $e\in \pE$. Once again,
(\ref{eq:rd4}) becomes trivial when $\sE=\pE$.

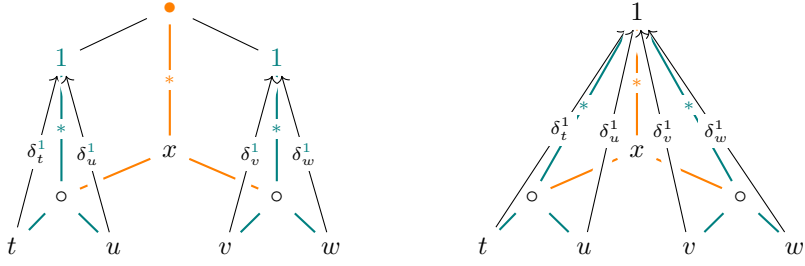
\begin{figure}[t]
  \centering

\begin{tikzcd}[column sep= .2cm, row sep=.2cm]
&&&\scomp \ar[dash,dll]\ar[dash,drr]&&&\\
&\pe &&&& \pe &\\
 &&  &\phantom{x}&  && \\
&&&x\ar[dash,line width=.8pt,orange,uuu, "\sconv" description] \ar[dash,line width=.8pt,orange,dll]\ar[dash, line width=.8pt,orange,drr]&&&\\
&\circ\ar[dash,line width=.8pt,teal,uuu, "\pconv" description] \ar[dash,line
width=.8pt,teal,dl]\ar[dash,line width=.8pt,teal,dr]&&&&
\circ\ar[dash,line width=.8pt,teal,uuu, "\pconv" description] \ar[dash,line width=.8pt,teal,dl]\ar[dash,line width=.8pt,teal,dr]&\\
t\arrow[uuuur, "\delta^{\pe}_t" description]&& u\ar[crossing over, uuuul,
"\delta^{\pe}_u" description] && v\ar[crossing over, uuuur, "\delta^{\pe}_v"
description] && w\ar[uuuul, "\delta^{\pe}_w" description]
\end{tikzcd}
\qquad\qquad
\begin{tikzcd}[column sep= .2cm, row sep=.2cm]
&&&1&&&\\
&&&\phantom{x}&&&\\
 &&  &\phantom{x}&  && \\
&&&x\ar[dash,line width=.8pt,orange,uuu, "\sconv" description] \ar[dash,line width=.8pt,orange,dll]\ar[dash, line width=.8pt,orange,drr]&&&\\
&\circ\ar[dash,line width=.8pt,teal,uuuurr, "\pconv" description] \ar[dash,line
width=.8pt,teal,dl]\ar[dash,line width=.8pt,teal,dr]&&&&
\circ\ar[dash,line width=.8pt,teal,uuuull, "\pconv" description] \ar[dash,line width=.8pt,teal,dl]\ar[dash,line width=.8pt,teal,dr]&\\
t\arrow[uuuuurrr, "\delta^{1}_t" description]&& u\ar[crossing over, uuuuur,
"\delta^{1}_u" description] && v\ar[crossing over, uuuuul, "\delta^{1}_v"
description] && w\ar[uuuuulll, "\delta^{1}_w" description]
\end{tikzcd}

\caption{Degeneracy condition (\ref{eq:d4}) in the presence of units.}
  \label{fig:delta-diagram2}
\end{figure}

\begin{corollary}\label{P:correspondence23-cor}
  Let $X$ be a unital relational bi-magma satisfying
  $\sE=\pE\neq \emptyset$ and $Q$ a unital bi-prequantale satisfying
  $\se=\pe\neq 0$. Then $(I_k)$ holds in $Q^X$ if and only if $(I_k)$
  holds in $Q$ and $(RI_k)$ holds in $X$, for each $1\le k\le 7$.
\end{corollary}

In the only-if directions, Functions $\delta_x$ can now be used. This
leads to a simpler relationship between deltas and ternary relations
than in Lemma~\ref{P:delta-rel}.
\begin{corollary}\label{P:delta-rel-unital-cor}
  Let $X$ be a relational magma and $Q$ a unital prequantale with
  $1\neq 0$. Then
  \begin{equation*}
    R^x_{yz} \Leftrightarrow (\delta_y \ast \delta_z)\, x = 1.
  \end{equation*}
\end{corollary}
\noindent It is therefore compelling to see $\bool$ as the sublattice
over $\{0,1\}$ of $Q$ and simply write
$R^x_{yz} = (\delta_y\ast\delta_z)\, x$ or even
$(f\ast g)\, x = \Sup_{y,z} f\, y \bullet g\, z \bullet R^x_{yz}$.
Figure~\ref{fig:delta-diagram2} shows how the presence of units
affects the right-hand term in (\ref{eq:ri7}).

Next we present a correspondence result for relational units that is
useful in Section~\ref{S:interchange-quantales}.

\begin{lemma}\label{P:unit-rel-conv}
Let $X$ be a unital bi-magma and $Q$ a unital bi-prequantale.
\begin{enumerate}
\item If $\pE \subseteq \sE$ in $X$ $\pe\le \se$ in $Q$, then $\pid
  \le \sid$ in $Q^X$.
\item If $\pid \le \sid$ in $Q^X$ and $\pe\neq 0$ in $Q$, then $\pE \subseteq \sE$ in
  $X$.
\item If $\pid \le \sid$ in $Q^X$ and $\pE\neq \emptyset$ in $X$, then $\pe\le \se$ in $Q$.
\end{enumerate}
\end{lemma}
\begin{proof}~
\begin{enumerate}
\item Let $\pE\subseteq \sE$ and $\pe\le \se$. Then $\sid\, x = 0 \Leftrightarrow \neg \sE^x \Rightarrow \neg \pE^x
  \Leftrightarrow \pid\, x = 0$ and therefore $\pid\le \sid$.
 \item Let $\pid\le \sid$. If $\pE^x$, then $0\neq \pe =\pid\, x \le
   \sid\, x$ and therefore $\sE^x$.
 \item Let $\pid\le \sid$ and $\pE^x$.  Then $\pe= \pid\, x \le \sid\, x \le \se$.\qedhere
\end{enumerate}
\end{proof}
\begin{corollary}\label{P:unit-rel-conv-cor}
  Let $X$ be a unital bi-magma with $\pE\neq \emptyset$ and $Q$ a
  unital bi-prequantale with $\pe\neq 0$. Then $\pid
  \le \sid$ in $Q^X$ if and only $\pE \subseteq \sE$ in $X$ and $\pe\le \se$ in $Q$.
\end{corollary}
Because of the symmetry in the definitions of unital bi-magmas and
bi-prequantales, Lemma~\ref{P:unit-rel-conv} and
Corollary~\ref{P:unit-rel-conv-cor} hold with colours swapped. We do
not spell them out explicitly.

The correspondences between interchange laws can be specialised to
obtain the associativity laws for a quantale.  The relational
interchange law (\ref{eq:ri3}),
$\exists y.\ \srel^x_{uy} \wedge \prel^y_{vw} \Rightarrow \exists y.\
\srel^y_{uv} \wedge \prel^x_{yw}$,
becomes the relational semi-associativity law
$\exists y. R^x_{uy} \wedge R^y_{vw}\Rightarrow \exists y.\ R^y_{u\,
  v} \wedge R^x_{y\, w}$
when colours are switched off; (\ref{eq:ri4}) translates into the
opposite implication.  Similarly, the interchange laws (\ref{eq:i3})
and (\ref{eq:i4}),
$a\scomp (b\pcomp c)\le (a\scomp b)\pcomp c$ and
$(a\pcomp b)\scomp c\le a \pcomp (b\scomp c)$, become the
semi-associativity laws
$a\bullet (b\bullet c)\le (a\bullet b)\bullet c$ and
$(a\bullet b)\bullet c\le a \bullet (b\bullet c)$.  This yields the
following corollary to Proposition~\ref{P:correspondence1},
\ref{P:correspondence2} and \ref{P:correspondence3}.

\begin{corollary}\label{P:assoc-correspondence1}
  Let $X$ be a relational magma and $Q$ a prequantale.
  \begin{enumerate}
  \item If $X$ is relationally associative and $Q$ associative, then
    $Q^X$ is associative.
\item If $Q^X$ is associative and some $a,b,c\in Q$ satisfy $a\bullet
  (b\bullet c)\neq 0\neq (a\bullet b) \bullet c$, then $X$ is
  relationally associative.
\item If $Q^X$ is associative and some $x,y,z,u,v,w\in X$ satisfy
  $R^x_{uz}$, $R^x_{yw}$ $R^z_{vw}$ and $R^y_{uv}$, then $Q$ is associative.
  \end{enumerate}
\end{corollary}
\noindent Similar correspondences between semi-associativity laws are
straightforward, but not as interesting for our purposes. In the
presence of units, Corollary~\ref{P:assoc-correspondence1}
simplifies further.

\begin{corollary}\label{P:assoc-correspondence2}
  Let $X$ be a unital relational magma satisfying $E\neq \emptyset$
  and $Q$ a unital prequantale satisfying $1\neq 0$.  Then $Q^X$ is
  associative if and only if $X$ is relationally associative and $Q$
  is associative.
\end{corollary}

The correspondences between interchange laws can also be specialised
to the commutativity law for a quantale.  The relational interchange
law (\ref{eq:ri2}), $\srel^x_{uv} \Rightarrow \prel^x_{vu}$,
specialises to $R^x_{uv} \Rightarrow R^x_{vu}$ when colours are
switched off while the interchange law (\ref{eq:i2}), $a\scomp b \le
b\pcomp a$, becomes $a\bullet b\le b\bullet a$.  This yields another
corollary to Proposition~\ref{P:correspondence1},
\ref{P:correspondence2} and \ref{P:correspondence3}.

\begin{corollary}\label{P:comm-correspondence}
Let $X$ be a relational magma and $Q$ a prequantale.
\begin{enumerate}
\item If $X$ is relationally commutative and $Q$ abelian, then $Q^X$
  is abelian.
\item If $Q^X$ is abelian and there exist $a,b\in Q$ with $a\bullet b\neq
  0$, then $X$ is relationally commutative.
\item If $Q^X$ is abelian and there exist $x,y,z\in X$ with $R^x_{x\, z}
$, then $Q$ is abelian.
\end{enumerate}
\end{corollary}
\noindent In the presence of units, this corollary simplifies further.
\begin{corollary}\label{P:comm-correspondence-unit}
  Let $X$ be a unital relational magma satisfying $E\neq \emptyset$
  and $Q$ a unital quantale satisfying $1\neq 0$. Then $Q^X$ is
  abelian if and only if $X$ is relationally commutative and $Q$ abelian.
\end{corollary}


\section{Relational Interchange Monoids and Interchange
  Quantales}\label{S:interchange-quantales}

We now start shifting the focus from correspondence theory to
construction recipes for quantales with interchange laws.  To
avoid nondegeneracy conditions, we assume that relational magmas and
quantales are unital and impose an order between units:
$\emptyset\neq \pE\subseteq \sE$ and $0\neq \pe\le \se$.

Yet first we prove a weak variant of the classical Eckmann-Hilton
argument~\cite{EckmannH62}. It shows that if a unital bi-magma, a set equipped with
composition $\scomp$ and unit $\se$ and composition $\pcomp$ with unit
$\pe$, satisfies the strong interchange law
$(a\pcomp b)\scomp (c\pcomp d) = (a\scomp c)\pcomp (b\scomp d)$, then
$\se=\pe$, $\scomp$ and $\pcomp$ coincide, and they are associative and
commutative. We show how these properties change if strong interchange
is weakened to (\ref{eq:i7}).  This of course requires ordered
bimagmas, in which the underlying set is partially ordered and the two
compositions preserve the order in both arguments.

\begin{lemma}[weak Eckmann-Hilton]\label{P:weak-eh}
  Let $(S,\le,\scomp,\pcomp,\se,\pe)$ be an ordered bimagma in which
  (\ref{eq:i7}) holds. Then $\se \le \pe$, and
  (\ref{eq:i1})--(\ref{eq:i6}) hold whenever $\pe \le \se$.
\end{lemma}
The proofs, like the classical Eckmann-Hilton ones, substitute $\se$ and
$\pe$ in (\ref{eq:i7}) and are straightforward. Analogous results hold for relational
bimagmas because of the various correspondence results in the
previous section and Lemma~\ref{P:weak-eh}.


\begin{lemma}\label{P:rel-interchange-redundant}
  Let $(X,\srel,\prel,\sE,\pE)$ be a unital relational bimagma in which (\ref{eq:ri7})
  holds. Then $\sE\subseteq \pE$, and (\ref{eq:ri1})--(\ref{eq:ri6}) hold
  whenever $\pE\subseteq \sE$.
\end{lemma}
\begin{proof}
First, for all $e\in S$, and with (\ref{eq:ri7}) in the fourth step, 
\begin{align*}
\sE^e
&\Leftrightarrow \sE^e \land \srel^e_{ee}\\
& \Leftrightarrow \sE^e \land \exists x,y,e',e''.\ \pE^{e'}\land \prel^x_{e'e} \land
  \srel^e_{xy}\land \prel^y_{ee''}\\
& \Rightarrow \exists e', e''.\ \sE^e \land \pE^{e'}\land \prel^e_{e'e} \land
  \srel^e_{ee}\land \prel^e_{ee''}\\
& \Rightarrow \exists e',e''.\ \sE^e \land \pE^{e'}\land \srel^e_{e'e} \land
  \prel^e_{ee}\land \srel^e_{ee''}\\
& \Rightarrow \sE^e \land \pE^e\land \srel^e_{ee} \land
  \prel^e_{ee}\land \srel^e_{ee}\\
&\Rightarrow \pE^e.
\end{align*}
Second, let $\pE\subseteq \sE$ and assume (\ref{eq:ri7}). Then
 \begin{align*}
  \exists y.\ \prel^y_{uv} \wedge \srel^x_{yw}
&\Leftrightarrow \exists e,y.\ \prel^y_{uv} \wedge \srel^x_{yw}\wedge
  \prel^w_{ew} \wedge \pE^e\\
&\Leftrightarrow \exists e,y,z.\ \prel^y_{uv} \wedge \srel^x_{yz}\wedge
  \prel^z_{ew} \wedge \pE^e\\
&\Rightarrow \exists e,y, z.\ \srel^y_{ue} \wedge \prel^x_{yz}\wedge
  \srel^z_{vw}\wedge \sE^e\\
&\Leftrightarrow \exists e,z.\ \srel^u_{ue} \wedge \prel^x_{uz}\wedge
  \srel^z_{vw}\wedge \sE^e\\
&\Leftrightarrow \exists z. \prel^x_{u\, z}\wedge \srel^z_{v\, w}
\end{align*}
proves (\ref{eq:ri6}). The proofs of (\ref{eq:ri1}), (\ref{eq:ri2}),
(\ref{eq:ri3}) and (\ref{eq:ri5}) from (\ref{eq:ri7}) are similar, and
left to the reader.
\end{proof}

\begin{definition}\label{D:relational-monoids}~
\begin{enumerate}
\item A \emph{relational semigroup} $(X,R)$ is a set $X$ equipped with
  a relationally associative ternary relation $R$.
\item A \emph{relational monoid} is a relational semigroup $(X,R)$
  with a set $E\subseteq X$ of relational units for $R$.
\item A \emph{relational interchange monoid} is a structure
  $(X,\srel,\prel,E)$ such that $(X,\srel,E)$ and $(X,\prel,E)$
  are relational monoids and the relational interchange
  law (\ref{eq:ri7}) holds.
\end{enumerate}
\end{definition}
\begin{definition}
  A (unital) \emph{interchange quantale} is a structure
  $(Q,\le,\scomp,\pcomp,1)$ such that $(Q,\le,\scomp,1)$ and
  $(Q,\le,\pcomp,1)$ are (unital) quantales, and the
  interchange law (\ref{eq:i7}) holds.
\end{definition}
In  light of Lemma~\ref{P:weak-eh} and
\ref{P:rel-interchange-redundant}, we always assume that relational
interchange monoids and interchange quantales have one single unit
that is shared between the relations and compositions, respectively.
The following result then summarises these two lemmas.
\begin{corollary}\label{P:relinterchange-redundant}~
\begin{enumerate}
\item In every relational interchange monoid,
  (\ref{eq:ri1})--(\ref{eq:ri6}) are derivable.
\item In every unital interchange quantale, (\ref{eq:i1})--(\ref{eq:i6})
  are derivable.
\end{enumerate}
\end{corollary}

The correspondence results from
Section~\ref{S:correspondence-interchange} and
\ref{S:further-correspondences} can now be summarised in terms of
interchange monoids and interchange quantales as well.

\begin{theorem}\label{P:interchange-quantale-correspondence}~
\begin{enumerate}
\item If $X$ is a relational interchange monoid and $Q$ an interchange
  quantale, then $Q^X$ is an interchange quantale.
\item If $Q^X$ is an interchange quantale and $1\neq 0$, then $X$ is a
  relational interchange monoid.
\item If $Q^X$ is an interchange quantale and $E\neq \emptyset$, then
  $Q$ is an interchange quantale.
\end{enumerate}
\end{theorem}
\begin{proof}
The correspondence for associativity and units in the subquantales is
given by Corollary~\ref{P:assoc-correspondence1} and
Proposition~\ref{P:correspondence-units}; that for interchange between
the subquantales by Propositions~\ref{P:correspondence1},
\ref{P:correspondence2} and \ref{P:correspondence3}.
\end{proof}
Theorem~\ref{P:interchange-quantale-correspondence} shows that, up to
mild nondegeneracy assumptions, all interchange quantales of type
$X\to Q$ are obtained from a relational interchange monoid on $X$ and
an interchange quantale $Q$. To build such quantales, one should
therefore look for relational interchange monoids, and the
advantage is that fewer properties need to be checked.

Interchange quantales generalise concurrent quantales and are strongly
related to concurrent Kleene algebras~\cite{HMSW11}. The difference is
that here we do not assume ``parallel composition'' $\pcomp$ to be
commutative. Yet Theorem~\ref{P:interchange-quantale-correspondence}
adapts easily to the commutative case. For a concurrent quantale in
$Q^X$, an interchange monoid $X$ with relationally commutative $\prel$
and an interchange quantale $Q$ with commutative $\pcomp$ is needed. A
variant of Theorem~\ref{P:interchange-quantale-correspondence} then
follows from Corollary~\ref{P:comm-correspondence-unit} and
\ref{P:correspondence23-cor}. In particular, the nondegeneracy
assumptions simplify to non-triviality assumptions for units and unit
sets.


\section{Duality for Powerset Quantales}\label{S:duality}

An interesting specialisation of
Theorem~\ref{P:interchange-quantale-correspondence} is the case of
$Q=\bool$, which forms an interchange quantale with both
compositions being meet and both units of composition $1$. In
particular, in the booleans, $0\neq 1$.  The interchange law
(\ref{eq:i7}) holds trivially because
$(w\wedge x)\wedge (y\wedge z) = (w \wedge y)\wedge (x\wedge z)$ in
any semilattice by associativity and commutativity of meet.
\begin{corollary}\label{P:complex_cor}
  $\bool^X\cong \pow\, X$ is an interchange quantale if
  and only if $X$ is a relational interchange monoid.
\end{corollary}

In this case, by Corollary~\ref{P:delta-rel-unital-cor} and Lemma~\ref{P:delta-rel},
$\srel^x_{yz} \Leftrightarrow (\delta_y \sconv \delta_z)\, x = 1$,
$\prel^x_{yz} \Leftrightarrow (\delta_y \pconv \delta_z)\, x = 1$ and
likewise for the other relational nondegeneracy conditions.

More interestingly, as a powerset boolean algebra, $\bool^X$ is
complete and atomic---a CABA---and a well known duality holds.  The
work of the Tarski school~\cite{JonssonT51} and Goldblatt~\cite{Goldblatt} shows
that categories of CABAs with $n$-ary operators and relational
structures with $n+1$-ary relations are dually equivalent.  Atomic
boolean prequantales are CABAs with a binary operator; relational
magmas are relational structures with a ternary relation.  Morphisms
in the category $\mathsf{ABP}$ of atomic boolean (pre)quantales
preserve sups, complementation and composition. A morphism $\rho$
between relational magmas $(X,R)$ and $(X',S)$ must satisfy
\begin{equation*}
R^x_{yz} \Rightarrow S^{(\rho\, x)}_{(\rho\, y)(\rho\, z)}
\end{equation*}
for all $x,y,z\in X$. A morphism is \emph{bounded} if, for all
$x,y,z\in X$,
\begin{equation*}
  S^{(\rho\, x)}_{yz} \Rightarrow \exists u,v\in X.\ \rho\, u =
  y\wedge \rho\, v= z \wedge R^x_{uv}.
\end{equation*}
The morphisms in the category $\mathsf{RM}$ of relational
magmas are assumed to be bounded.

Next we summarise this duality between categories.  With every atomic
boolean prequantale $Q$ one associates its dual relational
structure---its atom structure---$Q_+= \At\,Q$ by defining the ternary
relation $R$ in $Q$, as in Example~\ref{ex:memotrees}, by
\begin{equation*}
  R^\alpha_{\beta\gamma} \Leftrightarrow \alpha\le \beta\bullet \gamma
\end{equation*}
for all $\alpha,\beta,\gamma\in Q_+$. With very morphism
$\varphi:Q\to Q'$ in $\mathsf{ABP}$ one associates the map
$\varphi_+:\At\, Q'\to Q$ defined by
\begin{equation*}
  \varphi_+\, \beta = \bigwedge\{a\in Q\mid \beta \le \varphi\,
  a\}.
\end{equation*}
It is easy to check that $\varphi_+$ maps atoms in $Q'$ to atoms ind
$Q$.

Conversely, with every relational magma $(X, R)$ one associates its
dual convolution prequantale---its complex
algebra---$X^+= (\pow\, X,\subseteq,\ast)$. With every bounded
morphism $\rho:X\to X'$ one associates the contravariant powerset (or
preimage) functor $\rho^+: \pow\, X'\to \pow\, X$. It is defined, for
all $B\in X'$, by
\begin{equation*}
\rho^+\, B = \{x\in X\mid \rho\, x \in B\}.
\end{equation*}
In this context, our function $\delta:X\to X\to\bool$ is isomorphic to
the function $\eta:X\to \pow\, X$ defined by $\eta=\{-\}$.
Then $(\delta_y \ast \delta_z)\, x = 1 \Leftrightarrow x\in \{y\}\ast
\{z\}$ and hence $R^x_{y\, z} \Leftrightarrow x\in \{y\}\ast \{z\}$ by
Corollary~\ref{P:delta-rel-unital-cor}.

\begin{proposition}[\cite{JonssonT51,HenkinMonkTarski71}]\label{P:JT-dual}
Let $Q$ be an atomic boolean prequantale and $X$ a relational magma.
\begin{enumerate}
\item $Q\cong (Q_+)^+$ with isomorphism $\sigma:Q\to \pow\, (\At\, Q)$
  given by $\sigma\, a = \{\alpha\mid \alpha\le a\}$.
\item $X\cong
(X^+)_+$ with isomorphism $\eta:X\to \At\, (\pow\, X)$ given by
$\eta\, x = \{x\}$.
\end{enumerate}
\end{proposition}
To prove (1) one can use that any bijection $\varphi$ between the
atoms of two atomic boolean prequantales $Q$ and $Q'$ extends to an
isomorphism if and only if
$\alpha\le \beta\bullet \gamma\Leftrightarrow \varphi\, \alpha \le
\varphi\, \beta\bullet \varphi\, \gamma$
for all $\alpha,\beta,\gamma\in \At\, Q$.  The bijection
$\eta$ between atoms in $Q$ and atoms in $\pow\, Q$ satisfies this
condition, and it turns out that $\sigma$ is its extension.  For (2)
it is easy to check that the bijection $\eta$ is a relational magma
morphism.
\begin{proposition}[\cite{Goldblatt}]\label{P:Goldblatt-dual1}
The maps $(-)^+:\mathsf{RM}\to \mathsf{ABP}$ and
  $(-)_+:\mathsf{ABP}\to \mathsf{RM}$ are contravariant functors.
\end{proposition}
For $(-)^+$, one must show that $\rho^+$ preserves sups,
complementation and composition, for any bounded morphism $\rho$. The
first two properties follow from Stone's theorem for CABAs.  Proving
$\rho^+\, B_1 \ast \rho^+\, B_2 \subseteq \rho^+\, (B_1 \ast B_2)$ for
$B_1,B_2\in X'$ requires that $\rho$ is a morphism, the converse
inclusion that it is bounded.  Proving $(-)_+$ requires checking functoriality.
\begin{theorem}[\cite{Goldblatt}]\label{P:Goldblatt-dual2}
  The composites $(-)_+\circ (-)^+$ and $(-)^+\circ (-)_+$ are
  naturally isomorphic to the identity functors on the categories
  $\mathsf{ABP}$ and
  $\mathsf{RM}$, respectively. The two categories are dually
  equivalent.
\end{theorem}
The isomorphisms are $\sigma_Q: Q\mapsto (Q_+)^+$ and
$\eta_X:X\mapsto (X^+)_+$. Showing that these are components of a
natural transformations requires checking that the following diagrams commute.
\begin{equation*}
\begin{tikzcd}[column sep= 1cm, row sep=1cm]
Q \arrow{r}{\varphi}\arrow{d}{\sigma_Q}& Q'\arrow{d}{\sigma_{Q'}}\\
(Q_+)^+\arrow{r}{(\varphi_+)^+} & (Q'_+)^+
\end{tikzcd}
\qquad\qquad
\begin{tikzcd}[column sep= 1cm, row sep=1cm]
X \arrow{r}{\rho}\arrow{d}{\eta_X}& X'\arrow{d}{\eta_{X'}}\\
(X^+)_+\arrow{r}{(\rho^+)_+} & (X'^+)_+
\end{tikzcd}
\end{equation*}
\begin{question}
  Is there a Stone-type duality for non-atomic (boolean) quantales and
  arbitrary convolution algebras? 
\end{question}


\section{Interchange Kleene Algebras}\label{S:interchange-kas}

We mentioned in Section~\ref{S:summary} that in many classical
convolution algebras, including Rota's incidence algebras and the
formal power series of Scha\"utzenberger and Eilenberg's approach to
formal languages, the underlying set $X$ has a finite decomposition
property (Rota calls partial orders with this property \emph{locally
  finite}~\cite{Rota64}).  As infinite sups are then not needed to
express convolutions, one can specialise quantales to semirings and
Kleene algebras, and in particular concurrent Kleene algebras.  This
is the purpose of this section.

\begin{definition}
  A \emph{dioid} is a semiring $(S,+,\bullet,0,1)$ in which
  addition is idempotent.
\end{definition}
Hence $(S,+,0)$ is a sup-semilattice ordered by
$a\le b\Leftrightarrow a+b=b$ and least element $0$. Moreover $\bullet$
preserves $\le$ in both arguments.  A quantale can thus be seen as a
complete dioid.
\begin{definition}
  An \emph{interchange semiring} is a structure
  $(S,+,\scomp,\pcomp, 0,1)$ such that $(S,+,\scomp,0,1)$ and
  $(S,+,\pcomp,0,1)$ are dioids, and the
  interchange law (\ref{eq:i7}) holds.
\end{definition}
The six small interchange laws are of course derivable in this
setting.
\begin{definition}~
\begin{enumerate}
\item A \emph{Kleene algebra} is a dioid with a unary star operation
  $^\star$ that satisfies the unfold and induction axioms
\begin{equation*}
  1+a\bullet a^\star \le a^\star,\qquad c+a\bullet b\le b\Rightarrow a^\star
  \bullet b\le b,\qquad 1+a^\star \bullet a\le a^\star, \qquad c+b\bullet a
  \le b\Rightarrow b\bullet a^\star \le b.
\end{equation*}
\item An \emph{interchange Kleene algebra} is a structure
  $(K,+,\scomp,\pcomp, 0,1,^{\sstar},^{\pstar})$ such that
  $(K,+,\scomp,0,1,^{\sstar})$ and $(K,+,\pcomp,0,1,^{\pstar})$ are
  Kleene algebras and (\ref{eq:i7}) holds.
\end{enumerate}
\end{definition}
We write $(-)^\star$ instead of the usual $(-)^\ast$ to distinguish
the Kleene star from the convolution operation.

To translate Theorem~\ref{P:interchange-quantale-correspondence} into
the Kleene algebra setting all sups must be guaranteed to be finite.
This can be achieved by imposing that all functions have finite
support, or that the relations $\srel^x_{y\, z}$ and $\prel^x_{yz}$
are \emph{finitely decomposable}, that is, for each $x$ the sets
$\{(y,z)\mid \srel^x_{yz}\}$ and $\{(y,z)\mid \prel^x_{yz}\}$ are
finite. If this is the case we call the relational interchange monoid
\emph{finitely decomposable} as well.

\begin{theorem}\label{P:interchange-semiring-correspondence}
If $X$ is a finitely decomposable relational interchange monoid
  and $S$ an interchange semiring, then $S^X$ is an interchange
  semiring.
\end{theorem}
\begin{proof}~
In the construction of the convolution algebra on $S^X$ it can
  be checked that all sups remain finite.
\end{proof}

It is easy to generalise these results from dioids to proper semirings
that are ordered. We do not spell out the details. Beyond that it
seems interesting to extend
Theorem~\ref{P:interchange-semiring-correspondence} to interchange
Kleene algebras. First of all, every interchange quantale is an
interchange Kleene algebra, because $^{\sstar}$ and $^{\pstar}$ can be
defined explicitly in this setting using Kleene's fixpoint theorem:
$x^{\sstar} = \Sup_{k\in\mathbb{N}} x^{\textcolor{orange}{k}}$ and
$x^{\pstar} = \Sup_{k\in\mathbb{N}} x^{\textcolor{teal}{k}}$ satisfy
the star axioms, with powers defined recursively in the standard way
as $x^{\textcolor{orange}{0}} = 1$ and
$x^{\textcolor{orange}{i+1}} = x \scomp x^{\textcolor{orange}{i}}$ and
likewise for $x^{\textcolor{teal}{i}}$.

When infinite sups and the sup-preservation properties required for
Kleene's fixpoint theorem are not available, a different approach is
needed.  We have already shown~\cite{ArmstrongSW14} that
formal power series---functions of type $\Sigma^\ast\to K$, where
$\Sigma^\ast$ is the free monoid over the finite alphabet $\Sigma$ and
$K$ a Kleene algebra---form Kleene algebras. In this setting, the star
of a power series can be defined recursively~\cite{KuichS86} as
\begin{equation*}
  f^\star\, \varepsilon = (f\, \varepsilon)^\star, \qquad f^\star\, x = (f\, \varepsilon)^\star\bullet
  \sum_{y,z:x=y\cdot z,y\neq \varepsilon} f\, y \bullet f^\star\, z,
\end{equation*}
where $\sum$ indicates a finite sup. The verification of the star
axioms for power series uses structural induction over finite
words. Yet this is not applicable for general ternary relations.
Instead we use a notion of grading that been used
for arbitrary monoids by Sakarovitch~\cite{Sakarovitch03}.

The function $|-|:X\to\mathbb{N}$ is a \emph{grading} on the
relational monoid $(R,X,E)$ if
\begin{itemize}
\item $|x|>0$ for all $x\in X$ such that $x\notin E$,
\item $|x|=|y|+|z|$ whenever $R^x_{yz}$.
\end{itemize}
Then $(X,R,E)$ is \emph{graded} if there is a grading on $X$.
Thus, in a graded monoid $|e|=0$ if and only if $e\in E$.

\begin{proposition}\label{P:ka-correspondence-prop}
  If $(X,R,\{e\})$ is a graded, finitely decomposable, relational monoid and $K$ a Kleene
  algebra, then $K^X$ is a Kleene algebra with
\begin{equation*}
  f^\star\, e = (f\, e)^\star, \qquad f^\star\, x = (f\, e)^\star \bullet
  \sum_{y,z:R^x_{y
z},y\neq e} f\, y \bullet f^\star\, z.
\end{equation*}
\end{proposition}
\begin{proof}
  We need to check the unfold and induction axioms of Kleene
  algebra. First, it is well known that the unfold axiom
  $1+a^\star\bullet a \le a^\star$ is implied by the other axioms in
  any Kleene algebra, and can therefore be ignored.  Second, the
  axiom $c+a\bullet b\le b\Rightarrow a^\star \bullet c\le b$ follows
  from the simpler formula
  $a\bullet b\le b\Rightarrow a^\star \bullet b \le b$ and, by
  opposition, $c+b\bullet a\le b\Rightarrow c \bullet a^\star \le b$
  follows from $b\bullet a\le b\Rightarrow b \bullet a^\star \le b$, in
  any Kleene algebra~\cite{Kozen94}.  It
  thus remains to check that
\begin{equation*}
\id_e+f\ast f^\star\le f^\star,
\qquad
f \ast g \le g\Rightarrow f^\star \ast g \le g,
\qquad
g \ast f \le g\Rightarrow g \ast f^\star \le g
\end{equation*}
hold in the convolution algebra $K^X$.
\begin{enumerate}
\item $\id_e+f\ast f^\star = f^\star$. If $x=e$, then
$(\id_e + f \ast f^\star)\, e = 1 + (f\, e) \bullet (f\, e)^\star = (f\,
    e)^\star = f^\star \, e$.

 Otherwise, if $x\neq e$,
\begin{align*}
  (\id_e + f \ast f^\star)\, x &= \sum \left\{f\, y\bullet f^\star\, z \mid
  R^x_{yz}\right\}\\
&= f\, e \bullet f^\star \, x + \sum \left\{f\, y\bullet
  f^\star\, z \mid R^x_{yz}\land y \neq e\right\}\\
&= f\, e \bullet f^\star \, e \bullet \sum \left\{f\, y\bullet
  f^\star\, z \mid R^x_{yz}\land y \neq e\right\} + \sum \left\{f\, y\bullet
  f^\star\, z \mid R^x_{yz} \land y\neq e \right\}\\
&= \left(f\, e \bullet f^\star \, e + \id_E\, e\right) \bullet \sum \left\{f\, y\bullet
  f^\star\, z \mid R^x_{yz} \land y \neq e\right\}\\
&= \left(f\, e\right)^\star \bullet \sum \left\{f\, y\bullet f^\star\, z \mid
  R^x_{yz}\land y\neq e\right\}\\
&=f^\star \, x.
\end{align*}

\item
  $\left(\forall x.\ (f \ast g)\, x \le g\, x\right)\Rightarrow \left(\forall x.\
  (f^\star \ast g)\, x \le g\, x\right)$.  We proceed by induction on $|x|$.
\begin{enumerate}
\item Let $|x|=0$ and hence $x=e$. Then
  $(f^\star \ast g)\, e = (f\, e)^\star \bullet g\, e \le g\, e$
  follows from the assumption $f\, e\bullet g\, e\le g\, e$ and the
  first induction axiom of Kleene algebra.

\item Let $|x|>0$ and therefore $x\neq e$. Then, by the induction
  hypothesis, $\left(f\ast g\right)\, y\le g\, y$ holds for all $y$ with
  $|y|< |x|$.  In addition, the assumption implies that
  $\forall x,y,z.\ R^x_{yz} \Rightarrow f\, y\bullet g\, z\le g\, x$,
  from which
  $(f\, e)^\star\bullet g\, x = f^\star\, e \bullet g\, x \le g\, x$
  follows by star induction in $K$. With this property,
\begin{align*}
  (f^\star \ast g)\, x
 &= f^\star\, e \bullet g\, x +\sum \left\{f^\star\, e\bullet \sum\left\{f\,
   u \bullet f^\star\, v\mid R^y_{uv}\land u\neq e\right\} \bullet g\, z
   \mid R^x_{yz}\land y \neq e\right\}\\
&= f^\star\, e \bullet \left(g\, x + \sum\left\{(f\,
   u \bullet f^\star\, v) \bullet g\, z \mid \exists y.\ R^y_{uv}\land R^x_{yz}
  \land  u\neq e \wedge y\neq e\right\}\right)\\
&= f^\star\, e \bullet \left(g\, x + \sum\left\{ f\,
   u \bullet (f^\star\, v \bullet g\, z) \mid \exists y.\ R^x_{uy}\land
  R^y_{vz} \land u\neq e \land y\neq e\right\}\right)\\
&\le f^\star\, e \bullet \left(g\, x + \sum\left\{ f\, u \bullet (f^\ast \ast
  g)\, y \mid R^x_{uy}\land u\neq e \right\}\right)\\
&\le f^\star\, e \bullet \left(g\, x + \sum\left\{ f\, u \bullet g\, y\mid
  R^x_{uy} \right\}\right)\\
&\le f^\star\, e \bullet \left(g\, x + (f\ast g)\, x\right)\\
&= f^\star\, e \bullet \left(g\, x + g\, x\right)\\
&\le g\, x.
\end{align*}
The first step unfolds the definition of convolution and the Kleene
star in $K^X$. The second step applies distributivity laws in $K$; the
third one associativity in $X$ and $K$.  The fourth step introduces a
convolution as an upper bound, thus dropping the constraint $y\neq
e$. The fifth step applies the induction hypothesis to $y$. The
condition $u\neq e$ guarantees that $|y|<|x|$.  The sixth step applies
the assumption; the last step the derived property.
\end{enumerate}
\item $g \ast f \le g\Rightarrow g \ast f^\star \le g$
  follows by opposition from (2).\qedhere
\end{enumerate}
\end{proof}

\noindent The following theorem is then immediate.

\begin{theorem}\label{P:interchange-ka-correspondence}
If $X$ is a graded relational interchange monoid with unit $e$
  and $K$ an interchange Kleene algebra, then $K^X$ is an interchange
  Kleene algebra.
\end{theorem}

We have already discussed the relationship between interchange
quantales and concurrent quantales in
Section~\ref{S:interchange-quantales}, namely that concurrent
quantales are interchange quantales in which $\pcomp$ is commutative
and $\se=\pe$.  Similarly, concurrent semirings and concurrent Kleene
algebras are interchange semirings and interchange Kleene algebras
satisfying these two conditions.  It is then a trivial consequence of
Theorem~\ref{P:interchange-semiring-correspondence},
Corollary~\ref{P:comm-correspondence-unit} and Corollary
\ref{P:correspondence23-cor} that $S^X$ is a concurrent semiring if
$S$ is a concurrent semiring and $X$ a finitely decomposable
relational monoid. Similarly, by
Theorem~\ref{P:interchange-ka-correspondence} and these corollaries,
$K^X$ is a concurrent Kleene algebra if $K$ is a concurrent Kleene
algebra and $X$ a graded relational monoid.


\section{Weighted Shuffle Languages}\label{S:shuffle}

This extended example shows how weighted shuffle
languages~\cite{Handbook} can be constructed with our approach. Yet an
alternative view on relational interchange monoids is helpful.
Obviously, the sets $\pow\, (X\times X\times X)$ and $X\times X\to \pow\, X$
are isomorphic.
A ternary relation $R$ can thus be
seen as a \emph{multioperation} $\odot:X\times X\to \pow\, X$ defined
by
\begin{equation*}
  x \in y\odot z \Leftrightarrow R^x_{yz}.
\end{equation*}
It can be extended Kleisli-style to the operation
$\odot: \pow\, S\to\pow\, S\to \pow\, S$ defined,
for all $A,B\subseteq X$, by
\begin{equation*}
A\odot B = \bigcup \{x\odot y\mid x\in A\land y\in B\}.
\end{equation*}
It follows that
$(A\odot B)\, x = \Sup\{A\, y \wedge B\, z\mid R^x_{yz}\}$ if the
sets $A$ and $B$ are identified with their indicator functions, which
turns $\odot$ into a convolution.

The overloading of the multioperation $\odot$ and its extension allows
rewriting the relational interchange laws more compactly in algebraic
form.  Relational associativity becomes
$\{x\}\odot (x\odot z)= (x\odot y)\odot \{z\}$; the relational
interchange law (\ref{eq:ri7}) becomes
$(w\textcolor{teal}{\odot} x)\textcolor{orange}{\odot}
(y\textcolor{teal}{\odot} z) \subseteq (w\textcolor{orange}{\odot} y)
\textcolor{teal}{\odot} (x\textcolor{orange}{\odot} z)$.

Multisemigroups, multimonoids and other multialgebras have been
studied in mathematics for many
decades~\cite{Marty34,Krasner83,ConnesC10,KudryavtsevaM15}. In
computer science they appear in the semantics of separation
logic~\cite{GalmicheL06}.

  The shuffle of two words from an alphabet $\Sigma$ is obviously a
  multioperation
  $\|: \Sigma^\ast \to \Sigma^\ast \to \mathcal{P}\,\Sigma^\ast$. It
  can be defined recursively as
\begin{equation*}
  v\| \varepsilon = \{v\} = \varepsilon\| v,\qquad (av)
  \| (bw)= \{a\}\| (v\| (bw)) \cup \{b\} \| ((av)\|w),
\end{equation*}
where $a$ and $b$ are letters, $v$ and $w$ words and the extension
$\|:\pow\, X^\ast \to \pow\, X^\ast \to \pow\, X^\ast$ of $\|$ has
been tacitly used. It yields the shuffle or Hurwitz product
\begin{equation*}
A\| B=\bigcup\{x\| y\mid x\in A \land y\in B\}
\end{equation*}
for $A,B\subseteq \Sigma^\ast$ at language level.

To construct the quantale of $Q$-weighted shuffle languages using
Theorem~\ref{P:interchange-quantale-correspondence}(1) it remains to
check that the structure $M=(\Sigma^\ast,\srel,\prel,\{\varepsilon\})$
is a relational interchange monoid with shared unit
$\sE=\pE=\{\varepsilon\}$, where
$\srel^x_{yz} \Leftrightarrow x=y\cdot z$, for word concatenation
$\cdot$ and $\prel^x_{yz}\Leftrightarrow x \in y\|z$ for
shuffle.

It is of course straightforward to check that
$(\Sigma^\ast,\srel,\{\varepsilon\})$ is a relational monoid: it is in
fact isomorphic to the free monoid $(\Sigma^\ast,\cdot,\varepsilon)$
and checking the relational associativity and relational unit laws in
the first monoid amounts to checking their algebraic counterparts in
the second one. Checking that $(\Sigma^\ast,\prel,\{\varepsilon\})$ is a relational
monoid---or $(\Sigma^\ast,\|,\varepsilon)$ a multimonoid---and that
the relational interchange law (\ref{eq:ri7}) holds---or the
interchange law
$(w\|x)\cdot (y\| z) \subseteq (w\cdot y) \| (x\cdot z)$ with language
product $\cdot$ in the left-hand term and word concatenation $\cdot$
in the right-hand one---is a surprisingly tedious exercise and
requires nested
inductions.  

The result of this verification is summarised as follows.
\begin{lemma}\label{P:shuffle-quantale}
  $M$ is a relational interchange monoid with unit $\varepsilon$ and
  relationally commutative $\prel$.
\end{lemma}
\noindent The following corollary to
Theorem~\ref{P:interchange-quantale-correspondence}(1) is then automatic.
\begin{corollary}
  If $Q$ is an interchange quantale with unit $\se=\pe= 1$ and
  $\pcomp$ commutative, then $Q^M$ is an interchange quantale with
  $\pconv$ commutative and
\begin{equation*}
 (f\sconv g)\, x = \bigvee_{y,z:x= y\cdot z} f\, y \scomp g\, z,\qquad
 (f\pconv g)\, x = \bigvee_{y,z:x\in y\parallel z} f\, y \pcomp g\, z,\qquad
 \mathit{id}\, x = \delta_\varepsilon \, x.
\end{equation*}
\end{corollary}
\noindent The operation $\sconv$ is similar to the standard
convolution of formal power series, a $Q$-weighted generalisation of
the standard language product. The operation $\pconv$ generalises the
standard shuffle product $\|$ of languages to the $Q$-weighted
setting. Yet semirings or at least Kleene algebras are normally used
as weight-algebras.  A grading on words is needed, and in this
particular case the length of words can be used. It is then obvious
that $\Sigma^\ast_n$---the size of words of length $n$ is finite
whenever $\Sigma$ is finite. This yields the following corollary to
Theorem~\ref{P:interchange-ka-correspondence}.
\begin{corollary}\label{P:shuffle-ka}
  If $K$ is an interchange Kleene algebra with unit $1$ and $\pcomp$
  commutative, then $K^M$ is an interchange Kleene algebra with unit
  $\delta_\varepsilon$ and $\pconv$ commutative.
\end{corollary}
As we have shared units and a commutative shuffle operation, the
convolution algebras of weighted shuffle form in fact concurrent Kleene algebras.

Weighted languages are usually defined over semirings instead of
dioids. Instead of Kleene algebras one can then use star
semirings~\cite{Handbook}.  The Kleene star can then be defined on
$Q^M$ as before.  We conjecture that Corollary~\ref{P:shuffle-ka}
still holds for ordered star semirings, though we have not checked the
details.

Shuffle languages are widely used in the interleaving semantics of
concurrent programs. The finite transition and Aczel traces of the
rely-guarantee calculus~\cite{RoeverBH2001}, in particular, form concurrent
quantales, which suffices at least for the analysis of safety and
invariant properties.


\section{Partial Interchange Monoids}\label{S:pims}

Next we prepare for our second example, namely of digraphs under
serial and parallel composition.  It is then natural to consider these
compositions not as ternary relations, but as partial operations on
graphs.  This leads to more general notions of partial semigroups and
monoids. An approach to convolution with partial semigroups and
monoids has already been developed in~\cite{DongolHS16}.
\begin{definition}[\cite{DongolHS16}]
A \emph{partial monoid} is a structure
  $(S,\otimes, D, E)$ where $S$ is a set, $D\subseteq S\times S$ the
  domain of definition of the composition $\otimes :D\to S$, which is
  associative in the sense that
\begin{equation*}
  D\, x\, y \wedge D\, (x\otimes y)\, z \Leftrightarrow D\, y\, z\wedge
  D\, x\, (y\otimes z),\qquad D\, x\, y \wedge D\, (x\otimes y)\, z
  \Rightarrow x\otimes (y \otimes z) = (x\otimes y)\otimes z,
\end{equation*}
and $E\subseteq X$ is a set of units, which satisfy
\begin{equation*}
  \exists e\in E.\ D\, e\, x\wedge e\otimes x= x,\qquad \exists e\in E.\
  D\, x\, e\wedge x\otimes e= x,\qquad \forall e_1,e_2\in E.\ D\, e_1\, e_2
  \Rightarrow e_1= e_2.
\end{equation*}
\end{definition}
This definition captures the intuition of partiality that the
left-hand side of $x\otimes (y \otimes z) = (x\otimes y)\otimes z$ is
defined if and only if the right-hand side is, and, if either side is
defined, then the two sides are equal.  This notion of equality is
sometimes called \emph{Kleene equality}. Categories, monoids and the
interval algebras in Example~\ref{ex:incidence}, ordered pair algebras
in Example~\ref{ex:relations} and heaplet algebras in
Example~\ref{ex:sep-logic} all form partial monoids. Instead of the
unit axioms presented here we could equally use those of object-free
categories~\cite{MacLane98}. The precise relationship between partial
monoids and object-free categories is discussed in~\cite{CranchDS20}.

The relationship between partial and relational monoids is
straightforward. A relational monoid $(X,R,E)$ is \emph{functional} if
$R^x_{y\, z} = R^{x'}_{y\, z} \Rightarrow x = x'$ holds for all
$x,x'y,z\in X$.  With every functional relational monoid $(X,R,E)$ one
can then associate a partial monoid $(X,\otimes,D,E)$ with
$D\, y \,z \Leftrightarrow \exists x.\ R^x_{yz}$ and $y\otimes z$
being the unique $x\in X$ that satisfies $R^x_{yz}$---if $D\, y\, z$
is defined. We are particularly interested in the converse
construction.
\begin{lemma}[\cite{DongolHS17}]\label{P:pmonoid-relmonoid}
  If $(S,\otimes,D,E)$ is a partial monoid, then $(S,R,E)$ is a
  (functional) relational monoid with
  \begin{equation*}
R^x_{yz}\Leftrightarrow x=y\otimes z\land D\, y\, z.
\end{equation*}
\end{lemma}

Next we relate partial monoids with relational interchange monoids.
Expressing a variant of the interchange law (\ref{eq:i7}) in the
context of partial monoids requires an ordering on $S$. This motivates
the following definition.
\begin{definition}
  A \emph{preordered partial monoid} is a structure
  $(S,\preceq,\otimes,D,E)$ such that $(S,\preceq)$ is a preorder,
  $(S,\otimes,D,E)$ a partial monoid,
and $\otimes$ is order preserving in the
  sense that
  \begin{equation*}
    x\preceq y \wedge D\, z\, x \Rightarrow z \otimes x\preceq z\otimes y
    \wedge D\, z\, y,\qquad x\preceq y \wedge D\, x\, z \Rightarrow x
    \otimes z\preceq y\otimes z
    \wedge D\, y\, z.
  \end{equation*}
\end{definition}

Lemma~\ref{P:pmonoid-relmonoid} can then be generalised.
\begin{lemma}\label{P:ord-pmonoid-relsemigroup}
  Let $(S,\preceq,\otimes,D)$ be a preordered partial monoid. Then
  $(S,R)$ is a relational semigroup with
  \begin{equation*}
R^x_{yz}\Leftrightarrow x\preceq y\otimes z\land D\, y\, z.
\end{equation*}
  \end{lemma}
  \begin{proof}
   For relational associativity,
    \begin{align*}
      \exists y.\ R^x_{uy} \land R^y_{vw}
& \Leftrightarrow \exists y.\ x\preceq u\otimes y \land D\, u\, y
  \land y\preceq v\otimes w\land D\, v\, w\\
& \Leftrightarrow x\preceq u\otimes (v\otimes w)
  \land D\, v\, w \land D\, u\, (v\otimes w)\\
& \Leftrightarrow x\preceq (u\otimes v)\otimes w
  \land D\, u\, v \land D\, (u\otimes v)\, w\\
& \Leftrightarrow \exists y.\ D\, y \, w
  \land y \preceq u\otimes v \land D\, u\, v \land x\preceq y \otimes w\\
&\Leftrightarrow \exists y.\ R^y_{uv}\land R^x_{yw}.
    \end{align*}
\end{proof}
However the unit laws of preordered partial monoids need not translate
to relational semigroups.
\begin{lemma}\label{P:ord-pmonoid-relsemigroup-units}
  Let $(S,\preceq,\otimes,D)$ be a preordered partial monoid and
  $R^x_{yz}\Leftrightarrow x\preceq y\otimes z\land D\, y\, z$. Then
\begin{enumerate}
\item $\exists e\in E.\ R^x_{ex}$ and $\exists e\in E.\ R^x_{xe}$,
\item $\exists e\in E.\ R^x_{ey} \Rightarrow x\preceq y$ and $\exists e\in E.\ R^x_{ye} \Rightarrow x\preceq y$.
\end{enumerate}
\end{lemma}
In (2), it cannot generally be expected that $x=y$.  The relationship
$x\preceq y$ cannot be captured directly within relational semigroups
or monoids.

\begin{definition}
  A \emph{partial interchange monoid} is a structure
  $(S,\preceq,\sdot,\pdot,\sD,\pD,\sE\,\pE)$ such that
  $(S,\preceq,\sdot,\sD,\sE)$ and $(S,\preceq,\pdot,\pD,\pE)$ are
  preordered partial monoids, $\pE\subseteq \sE$ and the following
  interchange law holds:
  \begin{equation}
    \label{eq:pi7}
   \pD\, w\, x \land \sD\, (w\pdot x)\, (y\pdot z)\land \pD\, y\,
   z\Rightarrow \sD\, w\, y \land \pD\, (w\sdot y)\, (x \sdot z) \land \sD\, x\, z
\wedge (w \pdot x)\sdot (y\pdot z)\preceq (w\sdot y)\pdot (x\sdot z). \tag{pi7}
  \end{equation}
\end{definition}

In  light of Lemma~\ref{P:ord-pmonoid-relsemigroup}
and \ref{P:ord-pmonoid-relsemigroup-units} we cannot not expect to
relate partial interchange monoids directly with relational
interchange monoids. But the relationship is straightforward if we
forget relational units and restrict to a single monoidal unit.

\begin{lemma}\label{P:pim-small-interchange}
  Let $(S,\preceq,\sdot,\pdot,\sD,\pD,\sE,\pE)$ be a partial
  interchange monoid in which $\sE=\{e\}=\pE$. Then the following
  small interchange laws hold.
  \begin{enumerate}
  \item $\sD\, x \, y \Rightarrow \pD\, x\, y \land x \sdot y \preceq x \pdot y$,
\item $\sD\, x \, y \Rightarrow \pD\, y\, x \land x \sdot y \preceq y \pdot x$,
\item
  $\sD\, x\, (y\pdot z) \land \pD\, y\, z \Rightarrow \sD\, x\, y
  \land \pD\, (x \sdot y)\, x \land x\sdot (y \pdot z) \preceq (x\sdot
  y)\pdot z$,
\item $\pD\, x\, y \land \sD\, (x\pdot y)\, z \Rightarrow \pD\, x\, (y
  \sdot z) \land \sD\, y \, z\land (x\pdot y) \sdot z \preceq x \pdot (y\sdot z)$,
\item $\sD\, x \, (y\pdot z) \land \pD\, y\, z \Rightarrow \pD\, y\,
  (x\pdot z) \land \sD\, y\, z\, \land x\sdot (y \pdot z) \preceq y \pdot (x\sdot z)$,
\item $\pD\, x\, y\land \sD\, (x\pdot y)\, z \Rightarrow \sD\, x\,
  z\land \pD\, (x\sdot z)\, y \land (x\pdot y) \sdot z \preceq (x\sdot z) \pdot y$.
  \end{enumerate}
\end{lemma}
\begin{proof}
  We show (3) as an example. Suppose $\sD\, x\, (y\pdot z)$ and
  $\pD\, y\, z$. Then $\pD\, x\, e$ and
  $\sD\, (x\pdot e)\, (y\pdot z)$ and therefore $\sD\, x\, y$,
  $\pD\, (x\sdot y)\, (e\sdot z)$, $\sD\, e\, z$ and
  $(x\pdot e)\sdot (y \pdot z) \preceq (x\sdot y)\pdot (e\sdot z)$ by
  (\ref{eq:pi7}). Hence
  $x\sdot (y \pdot z) \preceq (x\sdot y)\pdot z$ by the unit laws of
  partial monoids. The other proofs are similar and left to the reader.
\end{proof}

With multiple units it seems necessary to require that parallel units
are sequential units for all elements, which is artificial.

From now on we call \emph{relational interchange semigroup} a
relational interchange monoid in which units may be absent, and the
small interchange laws (\ref{eq:ri1})-(\ref{eq:ri6}) hold in addition
to (\ref{eq:ri7}).

\begin{lemma}\label{P:pinterchangemonoid-rinterchangesemigroup}
  If $(S,\preceq,\sdot,\pdot,\sD,\pD,\{e\})$ is a partial interchange
  monoid, then $(S,\srel,\prel)$ is a relational interchange semigroup with
  $\srel^x_{yz}\Leftrightarrow x= y\sdot z\land \sD\, y\, z$
  and $\prel^x_{yz}\Leftrightarrow x\preceq y\pdot z\land \pD\, y\, z$.
  \end{lemma}
  \begin{proof}
    We need to check that (\ref{eq:pi7}) implies (\ref{eq:ri7}).
    \begin{align*}
      \exists y,z.\ \prel^y_{tu} \land \srel^x_{yz}\land \prel^z_{vw}
& \Leftrightarrow \exists y,z.\ y\preceq t\pdot u \land \pD\, t\,
  u\land x=y\sdot z\land \sD\, y\, z \land z\preceq v\pdot w \land
  \pD\, v\, w\\
&\Leftrightarrow x\preceq (t\pdot u)\sdot (v\pdot w) \land \pD\, t\,
  u\land \sD\, (t\pdot u)\, (v\pdot w) \land \pD\, v\, w\\
&\Rightarrow x\preceq (t\sdot v)\pdot (u\pdot w) \land \sD\, t\,
  v\land \pD\, (t\pdot v)\, (u\pdot w) \land \sD\, u\, w\\
&\Leftrightarrow \exists y,z.\ y = t\sdot v \land \sD\, t\, v\land
  x\preceq y\pdot z \land \pD\, y\, z
  \land z = u\sdot w \land \sD\, u\, w\\
& \Leftrightarrow \exists y,z.\ \srel^y_{tv} \land \prel^x_{yz} \land
  \srel^z_{uw}.
\qedhere
    \end{align*}
    This calculation does not depend on the presence of units.  Small
    interchange laws hold in $S$ by
    Lemma~\ref{P:pim-small-interchange}. These allow deriving the
    small relational interchange laws (\ref{eq:ri1})-(\ref{eq:ri6}) as
    in the proof of (\ref{eq:ri7}). Hence $S$ is a relational
    interchange semigroup.
  \end{proof}

It is easy to check that we
  could have used
  $\srel^x_{yz}\Leftrightarrow x\preceq y\sdot z\land D\, y\, z$
  instead of $\srel^x_{yz}\Leftrightarrow x = y\sdot z\land D\, y\, z$
  in the proof of
  Lemma~\ref{P:pinterchangemonoid-rinterchangesemigroup}. Using
  an equational encoding for $\prel$, however, would have broken the
  proof. The following example shows that even (\ref{eq:ri1}) would
  break if two equational encodings were used.

\begin{example}
  Consider the partial monoid over $\{a,b\}$ with $\sD= \{(a,a)\}$,
  $\pD =S\times S$, order and compositions defined by
  $b\pdot b=a\pdot b=b\pdot a= a\sdot a = b\prec a = a\pdot a$, and a
  suitable unit adjoined. The small interchange law
  $\sD\, x\, y \Rightarrow \pD\, x\, y \land x\sdot y\preceq x\pdot y$
  holds for all $x$ and $y$.  Let
  $\srel^x_{yz}\Leftrightarrow x=y\sdot z\land \sD\, y\, z$ and
  $\prel^x_{yz}\Leftrightarrow x=y\pdot z\land \pD\, y\, z$.  Then
  $\srel\not\subseteq \prel$, that is, $\srel^b_{aa}$ and
  $\neg \prel^b_{aa}$, because $\sD\, a\, a$, $b=a\sdot a$ and
  $b\neq a= a\pdot a$.\qed
\end{example}

Lemma~\ref{P:pinterchangemonoid-rinterchangesemigroup} yields the
following corollary to
Theorem~\ref{P:interchange-quantale-correspondence}(1).

\begin{corollary}\label{P:rel-interchange-monoid-quantale}
  If $S$ is a partial interchange monoid with unit $e$ and $Q$ an
  interchange quantale, then $Q^S$ is a non-unital interchange
  quantale with convolutions
\begin{equation*}
  (f\sconv g)\, x = \Sup_{y,z:x=y\sdot z} f\, y\scomp g\, z,\qquad
  (f\pconv g)\, x = \Sup_{y,z:x\preceq y\pdot z} f\, y\pcomp g\, z
\end{equation*}
that satisfies the small interchange laws
(\ref{eq:i1})--(\ref{eq:i6}) in addition to (\ref{eq:i7}).
\end{corollary}

Unitality fails in general because the unit $\id$ of $\sconv$ need
not be the unit of $\pconv$:
\begin{equation*}
(f\pconv \id)\, x = \bigvee\left\{f\, y \pcomp 1\mid \prel^x_{ye}\right\}=
\bigvee\left\{f\, y \mid x\preceq y\right\} \ge f\, x,
\end{equation*}
but not necessarily $f\pconv \id = f$, and similarly for
$\id\pconv f = f$. The retract $(Q,\le,\sconv)$ has unit $\id$;
only the retract $(Q,\le,\pconv)$ does not have $\id$ as a
unit. To obtain equality, and hence unital interchange quantales,
conditions on $f$ are needed.

A partial interchange monoid $(S,\sdot,\pdot,\{e\})$ is
\emph{positive} if $e$ is a minimal element of $S$ with respect to
$\preceq$. It is \emph{serially-decomposable} if
$x\preceq y_1\sdot y_2$ implies that there exists $x_1$, $x_2$ such
that $x=x_1\sdot x_2$, and $x_1\preceq y_1$ and $x_2\preceq y_2$.

\begin{lemma}\label{P:antitone-unital}
  Let $f$ be antitone, that is,
  $x\preceq y\Rightarrow f\, y\le f\, x$. Then
  $f\pconv \id = f = \id\pconv f$.
\end{lemma}
\begin{proof}
  $(f\pconv \id)\, x = \bigvee\left\{f\, y \pcomp 1\mid \prel^x_{ye}\right\}=
  \bigvee\{f\, y \mid x\preceq y\} = f\, x$.
  The $\le$-direction holds by antitonicity, the $\ge$-direction by
  the above calculation. The proof of $\id\pconv f = f$ is similar.
\end{proof}
To make $\id$ antitone it seems appropriate to require that $e$ is
minimal with respect to $\preceq$ and hence that the partial
interchange monoid is positive. We also need to check that $\sconv$
and $\pconv$ preserve antitonicity.

\begin{proposition}\label{P:pim-conv-quantale}
  Let $(S,\sdot,\pdot,\{e\})$ be a positive serially-decomposable
  partial interchange monoid and $Q$ and interchange quantale. Then
  the antitone functions in $Q^S$ form a (unital)
  interchange sub-quantale.
\end{proposition}
\begin{proof}
  Unitality follows from Lemma~\ref{P:antitone-unital}.  It remains to
  show that $\id$ is antitone and that $\sconv$ and $\pconv$ preserve
  antitonicity. The first fact follows from positivity.  For
  preservation of $\pconv$, suppose $x\preceq y$. Then
\begin{equation*}
  (f\pconv g)\, y = \bigvee \{f\, y_1 \pcomp f\, y_2\mid y \preceq y_1 \pdot
  y_2\} \le \bigvee \{f\, x_1 \pcomp f\, x_2\mid x \preceq x_1 \pdot
  x_2\} = (f\pconv g)\, x.
\end{equation*}
For preservation of $\sconv$, suppose once again $x\preceq y$. Then
\begin{equation*}
  (f\sconv g)\, y = \bigvee \{f\, y_1 \scomp f\, y_2\mid y = y_1 \sdot
  y_2\} \le \bigvee \{f\, x_1 \scomp f\, x_2\mid x = x_1 \sdot
  x_2\} = (f\sconv g)\, x
\end{equation*}
by $\sdot$-decompositionality.
\end{proof}

\section{Weighted Graph Languages}\label{S:graph-languages}

Our second extended example shows how weighted graph languages can be
constructed with our approach.  A partial interchange monoid structure
can be imposed on graphs in various ways.  Partiality arises because,
typically, the vertices of the graph operands are supposed to be
disjoint. Henceforce, we mean digraph when we say graph. Graphs with
undirected edges can be obtained from these in the obvious way.

Formally, we view graphs as binary relations on some set $X$.  Let
graphs $G_1$ and $G_2$ be disjoint, that is, they have disjoint vertex
sets: $V_{G_1}\cap V_{G_2}=\emptyset$. Their \emph{serial composition
  (complete join) and disjoint union (parallel composition)} are
defined as
\begin{equation*}
  G_1\cdot G_2=G_1\sqcup G_2\sqcup \, (V_{G_1}\times
V_{G_2}),\qquad
G_1\| G_2 = G_1\sqcup G_2,
\end{equation*}
where $\sqcup$ denotes disjoint union.  Both operations are
standard~\cite{CourcelleEngelfriet}. This turns graphs under
serial composition into partial monoids, and graphs under parallel
composition into partial abelian monoids.

A \emph{graph morphism} $\varphi:G_1\to G_2$ between graphs $G_1$ and
$G_2$ satisfies $(x,y)\in G_1\Rightarrow (f\, x,f\, y)\in G_2$.
A morphism $f$ is \emph{faithful}, or a \emph{graph embedding}, if
$(f\, x,f\, y)\in G_2$ implies $(x,y)\in G_1$.  A \emph{graph
  isomorphism} is a bijective (on vertices) graph embedding. We write
$G_1\cong G_2$ if there exists a graph isomorphism between $G_1$ and
$G_2$. We say that $G_1$ and $G_2$ are \emph{isomorphic} or have the
same \emph{graph type} if $G_1\cong G_2$ and call $G/{\cong}$ the
\emph{isomorphism class} or \emph{graph type} of $G$.

The \emph{subsumption relation} $\preceq$ between graphs, which is
defined by $G_1\preceq G_2$ if and only if there exists a bijective
(on vertices) graph morphism $\varphi:G_2\to G_1$, is a preorder. The
associated subsumption equivalence $\simeq$ need not coincide with
$\cong$, as will be explained in Section~\ref{S:graph-type-languages}.
We now fix any set $\mathcal{G}$ of (di)graphs that contains the empty
graph $\varepsilon$ and is closed under serial and parallel
composition.

\begin{proposition}\label{P:graph-cmonoid}
  The structure $(\mathcal{G},\cdot,\|,\{\varepsilon\})$ forms a
  partial interchange monoid with commutative parallel composition and
  shared unit $\varepsilon$.
\end{proposition}
\begin{proof}
  First of all, the partial associativity and unit laws, partial
  commutativity of disjoint union as well as partial isotonicity of
  the two compositions must be shown.This is routine.  In the presence
  of a shared unit $\varepsilon$ it then remains to verify
  (\ref{eq:pi7}).  For this we need the following isotonicity property
  of cartesian products: $A\subseteq B$ implies
  $A\times C\subseteq B\times C$ and $C\times A\subseteq C\times B$.

  We only show that the weak interchange law
  $(G_1\| G_2)\cdot (G_3\| G_4)\preceq_r (G_1\cdot G_2)\|(G_2\cdot
  G_4)$
  holds and leave the remaining laws to the reader.  We use the
  identity function on the $G_i$ to construct the bijective
  morphism. We need to show that
  $V_{(G_1\| G_2)\cdot (G_3\| G_4)}= V_{(G_1\cdot G_3)\|(G_2\cdot
    G_4)}$
  and
  $(G_1\cdot G_1)\|(G_3\cdot G_4) \subseteq (G_1\| G_3)\cdot (G_2\|
  G_4)$
  as a relation. First, $V_{G_i\cdot G_j} = V_{G_i}\cup V_{G_j} = V_{G_i\| G_j}$ and
  therefore
  \begin{equation*}
    V_{(G_1\| G_2)\cdot (G_3\| G_4) }
= V_{G_1} \cup V_{G_2}\cup V_{G_3}\cup V_{G_4}
=V_{(G_1\cdot G_3)\|(G_2\cdot G_4)}.
  \end{equation*}
Second,
  \begin{align*}
 (G_1\cdot G_1)\|(G_3\cdot G_4) &= (G_1\cup G_2\cup V_{G_1}\times
    V_{G_2})\|(G_3\cup G_4\cup V_{G_3}\times V_{G_4})\\
 &= G_1\cup G_2\cup V_{G_1}\times
    V_{G_2}\cup G_3\cup G_4\cup V_{G_3}\times V_{G_4}\\
 &\subseteq G_1\cup G_3\cup G_2\cup G_4 \cup (V_{G_1}\cup V_{G_3})\times
    (V_{G_2}\cup V_{G_4})\\
&= (G_1\| G_3) \cup (G_2\| G_4) \cup V_{G_1\| G_3}\times
    V_{G_2\| G_4}\\
& = (G_1\| G_3)\cdot (G_2\| G_4).
  \end{align*}
\end{proof}
Lemma~\ref{P:pinterchangemonoid-rinterchangesemigroup} and
Corollary~\ref{P:rel-interchange-monoid-quantale} then imply that
weighted graph languages form interchange quantales up-to unitality of
the parallel quantale retract. But one can do better.
\begin{lemma}
  The partial interchange monoid $(\mathcal{G},\cdot,\|,\{\varepsilon\})$ is positive and serially decomposable.
\end{lemma}
\begin{proof}
  It is clear that $\varepsilon$ is an isolated point with respect to
  $\preceq$ and hence minimal. This proves positivity.  The proof of
  serial decomposability is intuitive, but somewhat tedious to spell
  out formally. Suppose $G\preceq G_1\cdot G_2$.  Then the vertices
  of $G_1$ and $G_2$ are disjoint and in addition to the arrows of
  $G_1$ and $G_2$ we have $V_{G_1}\times V_{G_2}$. Hence if
  $G\preceq G_1\cdot G_2$, then the arrows added by the bijective
  graph morphism $\varphi:G_1\cdot G_1 \to G$ must either be added to
  $G_1$ or to $G_2$, while $V_{G_1}\times V_{G_2}$ stays the
  same. There must thus be $G_1'\preceq G_1$ and $G_2'\preceq G$ such
  that $G=G_1'\cdot G_2'$.
\end{proof}
Proposition~\ref{P:pim-conv-quantale} then specialises as follows.
\begin{corollary}\label{P:graph-conv-algebra}
If $Q$ is an interchange quantale with unit $1$ and $\pcomp$
commutative, then $Q^\mathcal{G}$ is a (generally non-unital) interchange
quantale with $\pconv$ commutative and
\begin{equation*}
  (f\sconv g)\, x = \Sup_{y,z:x=y\cdot z} f\, y\scomp g\, z,\qquad
  (f\pconv g)\, x = \Sup_{y,z:x\preceq y\| z} f\, y\pcomp g\, z.
\end{equation*}
The subquantale of antitone functions in $Q^\mathcal{G}$ is unital.
\end{corollary}
Labels can be added to vertices ad libitum, which yields proper
weighted graph languages. Both the serial and the parallel composition
preserve order properties. Corollary~\ref{P:graph-conv-algebra} thus
specialises immediately to weighted partial orders.

Next we consider convolution algebras that are powerset liftings, that
is, $Q= \bool$. Then $f:\mathcal{G}\to \bool$ is a set indicator
function and we may write $x\in f$ instead of $f\, x$, identifying the
indicator function with the set it represents. Then
$(f\sconv g)\, x = \Sup\{f\, y\scomp g\, z\mid x = y\cdot z\}$
rewrites as
$x\in f\sconv g \Leftrightarrow \exists y,z.\ x=y\cdot z\land y\in
f\land z\in g$ and hence
\begin{equation*}
f\sconv g= \{y\cdot
z\mid y\in f \land z\in g\}.
\end{equation*}
Similarly,
\begin{equation*}
f\pconv g=\{x\mid x\preceq
y\|z \land y\in f \land z\in g\}=\{y\|
z\mid y\in f \land z\in g\}{\downarrow},
\end{equation*}
where $\downarrow$ denotes the down-closure with respect to
$\preceq$. Moreover, the antitonicity condition rewrites as $x\preceq y
\land f\, y \Rightarrow f\, x$, which is precisely $f=f{\downarrow}$,
that is, $f$ is a down-set with respect to $\preceq$.
\begin{corollary}
  The down-sets in $\pow\, \mathcal{G}$ form a unital
  interchange quantale.
\end{corollary}
Finally we consider the finite case, and obtain the following
corollary of Theorem~\ref{P:interchange-ka-correspondence}.
\begin{corollary}
  If $K$ is an interchange Kleene algebra with unit $1$ and
  $\mathcal{G}$ a partial interchange monoid of finite graphs, then
  the antitone functions in
  $K^\mathcal{G}$ form an interchange Kleene algebra.
\end{corollary}
This holds because any finite graph can be decomposed in finitely many
ways serially or parallelly into subgraphs. Once again, all results
specialise to partial orders, and in particular to labelled partial
orders, where vertices are labelled with letters from some alphabet.
Sets of partial orders in general, and labelled partial orders in
particular, are widely used in concurrency
theory~\cite{Grabowski,Vogler92} and the theory of distributed
systems~\cite{Lamport78}.


\section{Weighted Languages of Types of Finite
  Graphs}\label{S:graph-type-languages}

Many applications, including those in concurrency and distributed
systems, require isomorphism classes and hence types of graphs or
(labelled) partial orders. Lifting the results from
Section~\ref{S:graph-languages} to these is not entirely
straightforward. This is well known~\cite{Esik02}, but we spell out
details for the sake of completeness.

\begin{example}[\cite{Esik02}]
  Consider the infinite poset $(P,\le_P)$ with
  $P=\{p_{i,j}\mid i,j\in \mathbb{N}\land (i=0\lor j=0)\}$ and
  $p_{i,j}\le_P p_{k,l}$ if and only if $i = k = 0$ and $j \le l$, and
  the infinite poset $(Q,\le_Q)$ with
  $Q=\{q_{i,j}\mid i,j\in \mathbb{N}\land (i=0\lor i=1 \lor j=0)\}$
  and $q_{i,j}\le_P q_{k,l}$ if and only if $i = k = 0$ or $i=k=1$,
  and $j\le l$.

  Intuitively, $P$ consists of the disjoint union of the infinite
  chain formed by the $p_{0,j}$ and the elements $p_{i,0}$ with
  $i> 0$, whereas $Q$ consists of the disjoint union of the infinite
  chain formed by the $p_{0,j}$, the infinite chain formed by the
  $p_{1,j}$ and the elements $p_{i,0}$ with $i\ge 1$.

Define the functions $\varphi:P\to Q$ and $\psi:Q\to P$ by
\begin{equation*}
  \varphi\, p_{i,j} =
  \begin{cases}
    q_{0,j} & \text{ if } i=0,\\
q_{1,k} & \text{ if } i > 0 \land j = 2k+1,\\
q_{k,0} & \text{ if } i > 0 \land j = 2k,
  \end{cases}
\qquad
  \psi\, q_{i,j} =
  \begin{cases}
    p_{0,2k} & \text{ if } i=0,\\
p_{1,2k+1} &\text{ if } i=1,\\
p_{i-1,0} & \text{ if } i > 1.
  \end{cases}
\end{equation*}
Intuitively, $\varphi$ maps the chain in $P$ onto the first chain in
$Q$ and the isolated elements in $P$ alternatingly onto the second
chain and the isolated elements in $Q$, whereas $\psi$ maps the
elements of the two chains in $Q$ alternatingly onto the chain in $P$,
and isolated points in $Q$ onto isolated points in $P$. The morphisms
are shown in Figure~\ref{fig:graph-morphisms}.

\begin{figure}[t]
  \centering
\begin{equation*}
    \begin{tikzcd}[column sep =.1cm, row sep = -.1cm,shorten <=-3pt, shorten >=-3pt]
&\phantom{\circ}&&&&&&\phantom{\circ} \\
&\phantom{\circ}&&&&&&& \\
&\circ\ar[dash,dashed,uu] \ar[blue,mapsto,bend left=10,rrrrrr] &&&&&\phantom{\circ}&\circ\ar[dash,dashed,uu] \\
\phantom{\circ}&&&&&\phantom{\circ}&&\\
&\circ \ar[blue,mapsto,bend left=10,rrrrr] &&&&&\circ\ar[dash,dashed,uu] &\\
\circ\ar[dash,dashed,uu]\ar[blue,mapsto,bend left=10,rrrrr] &&&&&\circ\ar[dash,dashed,uu] &&\\
&\circ \ar[blue,mapsto,bend left=10,rrrrrr] &&&&&&\circ \\
\phantom{\circ}&&&&&&&\\
& \circ \ar[blue,mapsto,bend left=10,rrrrr] &&&&& \circ\ar[dash,uuuu] & \\
\circ\ar[dash,uuuu]\ar[blue,mapsto,bend left=10,rrrrr] &&&&& \circ\ar[dash,uuuu] && \\
& \circ \ar[blue,mapsto,bend left=10,rrrrrr] &&&&&&\circ\\
\phantom{\circ}&&&&&&&\\
& \circ \ar[blue,mapsto,bend left=10,rrrrr] &&&&&\circ\ar[dash,uuuu]& \\
\circ\ar[dash,uuuu]\ar[blue,mapsto,bend left=10,rrrrr] &\phantom{\circ}&\phantom{\circ}&\phantom{\circ}&\phantom{\circ}&\circ\ar[dash,uuuu] &&
\end{tikzcd}
\qquad\qquad
   \begin{tikzcd}[column sep =.1cm, row sep = -.1cm,shorten <=-3pt,
     shorten >=-3pt]
&&\phantom{\circ}&&&&&\phantom{\circ}\\
&\phantom{\circ}&&&&&\phantom{\circ}&\\
&&\circ\ar[dash,dashed,uu] \ar[blue,mapsto,bend left=10,rrrrr] &&&&&\circ\ar[dash,dashed,uu] \\
\phantom{\circ} &\circ\ar[dash,dashed,uu]\ar[blue,mapsto,bend
left=10,rrrrr] &&&&&\circ\ar[dash,dashed,uu] &\\
\phantom{\circ}&&&&&&&\\
\circ\ar[dash,dashed,uu] \ar[blue,mapsto,bend left=10,rrrrrr] &&&&&&\circ\ar[dash,uu] &\\
&&\circ \ar[blue,mapsto,bend left=10,rrrrr] &&&&&\circ\\
& \circ\ar[dash,uuuu]\ar[blue,mapsto,bend left=10,rrrrr]
&&&&&\circ\ar[dash,uu]&\\
\phantom{\circ}&&&&&&&\\
\circ\ar[dash,uuuu]\ar[blue,mapsto,bend left=10,rrrrrr] &&&&&&\circ\ar[dash,uu] &\\
&&\circ \ar[blue,mapsto,bend left=10,rrrrr]  &&&&&\circ\\
&\circ\ar[dash,uuuu]\ar[blue,mapsto,bend left=10,rrrrr]
&&&&&\circ\ar[dash,uu]&\\
\phantom{\circ}&&&&&&&\\
\circ\ar[dash,uuuu]\ar[blue,mapsto,bend left=10,rrrrrr] &&&\phantom{\circ}&\phantom{\circ}&\phantom{\circ}&\circ\ar[dash,uu] &
\end{tikzcd}
\end{equation*}
 \caption{Posets $P$ and $Q$ with bijective morphisms $\varphi$ in left
   diagram and $\psi$ in right diagram.}
  \label{fig:graph-morphisms}
\end{figure}
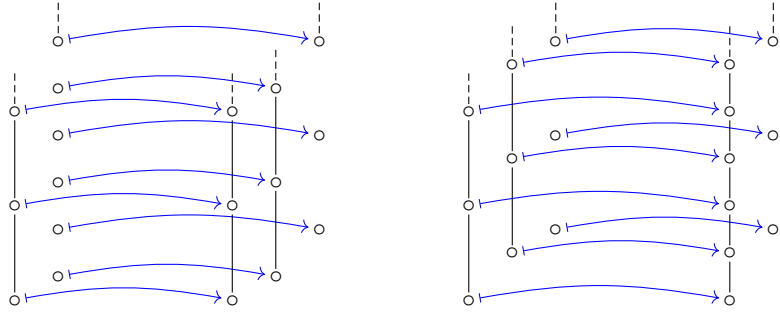

By construction, $\varphi$ and $\psi$ are both bijective and
order-preserving. Hence $P\preceq Q$ and $Q\preceq P$, but of course
neither $P=Q$ nor $P\cong Q$.\qed
\end{example}

At least in the finite case the situation is simpler. An explanation
requires two simple facts about groups.

\begin{lemma}\label{P:group-aux1}
  Let $G$ be the cyclic group generated by $x$ and let $x^i = x^j$ for
  some integers $i<j$.  Then $G=\{1,x,x^2,\dots x^{k-1}\}$, where
  $k=j-i$.
\end{lemma}
\begin{proof}
  The assumption implies that $x^i = x^i x^k$ for some $k$, and thus
  $x^k= 1$ by cancellation.  By cyclicity, every $g\in G$ is of the
  form $g=x^n$ for some $n\in\mathbb{N}$ that can be written $n=pk+q$
  for some $p,q\in\mathbb{N}$ with $q\le k-1$.  Hence
  $g=(x^k)^p x^q=1^p x^q = x^q$ for some $0\le q\le k-1$ or,
  equivalently, $g\in \{1,x,x^2,\dots x^{k-1}\}$.  Since every
  $x^n\in G$, this shows that $G=\{1,x,x^2,\dots x^{k-1}\}$.
\end{proof}
\begin{lemma}\label{P:group-aux2}
  Let $G$ be a finite cyclic group of order $n$ generated by $x$.
  Then $G=\{1,x,x^2,\dots, x^{n-1}\}$ and $x^n=1$.
\end{lemma}
\begin{proof}
  By the pigeonhole principle, there must be a minimal
  $j\in\mathbb{N}$, $j\le n$, such that $x^j=x^i$ for some $i\in\mathbb{N}$ with
  $i<j$. Hence the elements $1,x,x^2,\dots, x^{j-1}$ are pairwise
  distinct.  Then, by Lemma~\ref{P:group-aux1}, $j=n$, $i=0$ and $x^n= x^0=1$.
 \end{proof}
\begin{lemma}\label{P:refinement-fin-equiv}
  Let $G_1$ and $G_1$ be finite graphs such that $G_1\preceq G_2$ and
  $G_2\preceq G_1$.  Then $G_1\cong G_2$.
\end{lemma}
\begin{proof}
  By assumption there exists order preserving bijections
  $\varphi:G_2\to G_1$ and $\psi:G_1\to G_2$, hence
  $\chi=\psi\circ \varphi$ is an order preserving bijection on $G_1$.
  As $\chi$ can be seen as a group action on the finite set $V_1$, it
  generates a finite cyclic group.  Hence there is some
  $k\in\mathbb{N}$ such that $\chi^k=\mathit{id}_{V_1}$ by
  Proposition~\ref{P:group-aux2}. It then follows that $f$ is
  faithful: Suppose $(\varphi\, x,\varphi\, y)\in G_2$. Then
  $ x = \chi^k\, x = \chi^{k-1}(\psi(\varphi\, x))\to_R
  \chi^{k-1}(\psi(\varphi\, y)) = \chi^k\, y = y$.
  It follows that $\varphi$ is a graph isomorphism and $G_1\cong G_2$.
\end{proof}
A similar fact has been proved by \'Esik~\cite{Esik02}. We henceforth
restrict our attention to finite graphs.

Let $[G]=\{G'\mid G'\cong G\}$ denote the type of $G$.  We
extend the subsumption preorder $\preceq$ to equivalence classes by
$[G_1]\preceq [G_2]\Leftrightarrow G_1\preceq G_2$, overloading
notation.  This relation is well defined.
 \begin{lemma}\label{P:type-preorder-defined}
   Let $G_1'\cong G_1$, $G_1\preceq G_2$ and $G_2\cong G_2'$. Then
   $G_1'\preceq G_2'$.
 \end{lemma}
\begin{proof}
  Let $\varphi_1$ be the graph isomorphism of type $G_1\to G_1'$,
  $\varphi_2$ the graph isomorphism of type $G_2'\to G_2$ and $\psi$
  the bijective graph morphism of type $G_2\to G_1$. Then
  $\varphi_1\circ \psi\circ \varphi_2:G_2'\to G_1'$ is a bijective
  graph morphism as well. Hence $G_1' \preceq G_2'$.
\end{proof}
\begin{lemma}\label{P:type-refinement-fin-po}
The relation $\preceq$ is a partial order on $\mathcal{G}/{\cong}$ if all graphs
in $\mathcal{G}$ are finite.
\end{lemma}
\begin{proof}
  Reflexivity and transitivity for $\preceq$ on $\mathcal{G}/{\cong}$
  follows from reflexivity and transitivity of $\preceq$ on
  $\mathcal{G}$.  For antisymmetry, $[G_1]\preceq [G_2]$ and
  $[G_2]\preceq [G_1]$ imply $[G_1]=[G_2]$ for all
  $G_1,G_2\in\mathcal{G}$ by Lemma~\ref{P:refinement-fin-equiv}.
\end{proof}
Extending serial and parallel composition of graphs is standard:
$[G_1]\cdot [G_2]=\{G_1'\cdot G_2'\mid G_1'\cong G\wedge G_2'\cong
G\}$
and likewise for $[G_1]\| [G_2]$. It is also well known that both
compositions are well defined: if $G_1\cong G_1'$ and $G_2\cong G_2'$,
then $[G_1]\cdot[G_2]=[G_1']\cdot [G_2']$ and
$[G_1]\|[G_2]=[G_1']\| [G_2']$. By contrast to serial and parallel
compositions of graphs, those of graph types are total. Finally,
equivalence classes are closed with respect to serial and parallel
composition.

\begin{lemma}\label{P:type-compositions-closed}
For all $G_1,G_2\in\mathcal{G}$,
  \begin{enumerate}
  \item $[G_1\cdot G_2] = [G_1]\cdot [G_2]$,
\item $[G_1\| G_2] = [G_1]\|[G_2]$.
  \end{enumerate}
\end{lemma}
\begin{proof}
  $H\in [G_1\cdot G_2]$ if and only if $H\cong G_1\cdot G_2$. This is
  the case if and only if there are graphs
  $G_1'$ and $G_2'$ such that $H=G_1'\cdot G_2'$ and $G_1'\cong G_1$ and
  $G_2'\cong G_2$, which holds if and only if $H\in [G_1]\cdot [G_2]$. The proof for $\|$ is similar.
\end{proof}
\begin{proposition}\label{P:graph-type-interchange-monoid}
  The structure $(\mathcal{G}/{\cong},\cdot,\|,[\varepsilon])$ is an
  interchage monoid in which $\|$ is commutative, if all graphs in
  $\mathcal{G}$ are finite.
\end{proposition}
\begin{proof}
  The associativity, commutativity and unit laws are easy to check,
  noting that $[\varepsilon]=\{\varepsilon\}$.  For the interchange
  law
  $[(G_1\| G_2)\cdot (G_3\| G_4)]\preceq [(G_1\cdot G_3)\|(G_2\cdot
  G_4)]$,
  by definition of $\preceq$ on equivalence classes, it suffices to
  show that $(G_1\| G_2)\cdot (G_3\| G_4)\preceq (G_1\cdot G_3)\|(G_2\cdot
  G_4)$, which holds by Proposition~\ref{P:graph-cmonoid}.
\end{proof}

Proposition~\ref{P:graph-type-interchange-monoid} specialises
immediately to types of finite partial orders with serial and parallel
composition, which are known as partial words or pomsets in
concurrency theory---when vertex labels are
added~\cite{Grabowski,Gischer}. The instance of
Proposition~\ref{P:graph-type-interchange-monoid} for pomsets is due
to Gischer~\cite{Gischer}.

Because some compositions in $\mathcal{G}/{\cong}$ may results in the
empty set, the interchange monoid can have an annihilator $0$, that
is, an element for which $x\cdot 0=0= 0\cdot x$ and $x\| 0 = 0$ holds
for any element $x$.

The lifting to convolution algebras---interchange quantales, unital
interchange quantales, interchange Kleene algebras---then follows the
results of the previous section. The result that the powerset lifting
of finite pomsets yields concurrent semirings, interchange semirings
in which $\pcomp$ is commutative and $Q=\bool$, is due to
Gischer~\cite{Gischer}. Extensions to concurrent Kleene algebras and
concurrent quantales have been proved more recently~\cite{HMSW11}.


\section{Conclusion}\label{S:conclusion}

The results in this article support the construction of concurrent
quantales and Kleene algebras from relational structures, multimonoids
and partial monoids. They can be formalised easily in proof assistants
and applied in concurrency verification. In fact, the lifting from
ternary relations and partial monoids to quantalic convolution
algebras---without interchange laws---has already been formalised with
Isabelle/HOL~\cite{DongolGHS17}. Extending this to concurrency is left
for future work.

Another interesting avenue for research is the extension of Stone-type
duality to our constructions, building on work of Harding, Walker and
Walker for lattice-valued functions~\cite{HardingWW18}. Moreover, a
categorification of our approach will be published in a successor
paper.

Finally, we hope that our results will benefit the construction of
real-word graph-based models for concurrent and distributed systems,
and ultimately the design of programming languages and verification
tools for such systems.

\vspace{\baselineskip}

\textbf{Acknowledgement:} The authors would like to thank Tony Hoare
for discussions on models of concurrent Kleene algebras, and to
Brijesh Dongol and Ian Hayes for their collaboration on convolution
algebras and comments on early versions of this work. The second and
third author have been partially supported by EPSRC grant EP/R032351/1.


\bibliographystyle{alpha}
\bibliography{cka-conv2}

\end{document}